\documentclass[a4paper,UKenglish]{lipics-v2016}

\usepackage{microtype}
\usepackage{amsthm}

\def\romannum{\begingroup
  \def\theenumi{\textup{(\roman{enumi})}}%
  \def\p@enumi{}%
  \def\labelenumi{\theenumi}%
  \enumerate}

\newcommand{\ignore}[1]{}

\newcommand{\integerset}[0]{\ensuremath{\mathbb{N}}}

\newcommand{\restrict}{\ensuremath{\upharpoonright}}

\newcommand{\reactivelycomposable}{\ensuremath{\trianglelefteq}}

\newtheorem{proposition}{Proposition}
\newtheorem{conjecture}{Conjecture}
\newtheorem{claim}{Claim}

\begin{document}

\title{The complexity of quantified constraints using the algebraic formulation}

\author[1]{Catarina Carvalho}
\author[2]{Barnaby Martin}
\author[3]{Dmitriy Zhuk}

\affil[1]{School of Physics, Astronomy and Mathematics, University of Hertfordshire, U.K.}
\affil[2]{School of Engineering and Computing Sciences, Durham University, U.K.}
\affil[3]{Moscow State University, 119899 Moscow, Russia.}
\authorrunning{C. Carvalho et al.}

\Copyright{C. Carvalho et al..}

\subjclass{F.2.2 Nonnumerical Algorithms and Problems}
\keywords{Quantified Constraints, Computational Complexity, Universal Algebra, Constraint Satisfaction}

\EventEditors{John Q. Open and Joan R. Acces}
\EventNoEds{2}
\EventLongTitle{42nd Conference on Very Important Topics (CVIT 2016)}
\EventShortTitle{CVIT 2016}
\EventAcronym{CVIT}
\EventYear{2016}
\EventDate{December 24--27, 2016}
\EventLocation{Little Whinging, United Kingdom}
\EventLogo{}
\SeriesVolume{42}
\ArticleNo{23}

\maketitle


\begin{abstract}
Let $\mathbb{A}$ be an idempotent algebra on a finite domain. We combine results of Chen \cite{AU-Chen-PGP}, Zhuk \cite{ZhukGap2015} and Carvalho \mbox{et al.} \cite{LICS2015} to argue that if $\mathbb{A}$ satisfies the polynomially generated powers property (PGP), then QCSP$(\mathrm{Inv}(\mathbb{A}))$ is in NP. We then use the result of Zhuk to prove a converse, that if $\mathrm{Inv}(\mathbb{A})$ satisfies the exponentially generated powers property (EGP), then QCSP$(\mathrm{Inv}(\mathbb{A}))$ is co-NP-hard. Since Zhuk proved that only PGP and EGP are possible, we derive a full dichotomy for the QCSP, justifying the moral correctness of what we term the Chen Conjecture (see \cite{Meditations}).

We examine in closer detail the situation for domains of size three. Over any finite domain, the only type of PGP that can occur is switchability. Switchability was introduced by Chen in  \cite{AU-Chen-PGP} as a generalisation of the already-known Collapsibility \cite{hubie-sicomp}. For three-element domain algebras $\mathbb{A}$ that are Switchable, we prove that for every finite subset $\Delta$ of Inv$(\mathbb{A})$, Pol$(\Delta)$ is Collapsible. The significance of this is that, for QCSP on finite structures (over three-element domain), all QCSP tractability explained by Switchability is already explained by Collapsibility.

Finally, we present a three-element domain complexity classification vignette, using known as well as derived results. 
%
\end{abstract}


\section{Introduction}

A large body of work exists from the past twenty years on
applications of universal algebra to the computational complexity of
\emph{constraint satisfaction problems} (CSPs) and a number of
celebrated results have been obtained through this method. 
One considers the problem CSP$(\mathcal{B})$ in which it is asked whether an input
sentence $\varphi$ holds on $\mathcal{B}$, where $\varphi$ is
\emph{primitive positive}, that is using only $\exists$, $\land$ and $=$. The CSP is one of a wide class of model-checking problems obtained from restrictions of first-order logic. For almost every one of these classes, we can give a complexity
classification~\cite{DBLP:journals/corr/abs-1210-6893}: the two
outstanding classes are CSPs and its popular extension \textsl{quantified
  CSPs} (QCSPs) for positive Horn sentences -- where $\forall$ is also present -- which is used in
Artificial Intelligence to model non-monotone reasoning or
uncertainty \cite{nonmonotonic}. 

The outstanding conjecture in the area is that all finite-domain
CSPs are either in P or are NP-complete, something surprising given
these CSPs appear to form a large microcosm of NP, and NP itself is
unlikely to have this dichotomy property. This Feder-Vardi
conjecture \cite{FederVardi}, given more concretely in the algebraic
language in \cite{JBK}, remains unsettled, but is now known for
large classes of structures. It is well-known that the complexity classification for QCSPs embeds the classification for CSPs: if $\mathcal{B}+1$ is $\mathcal{B}$ with the addition of a new isolated element not appearing in any relations, then CSP$(\mathcal{B})$ and QCSP$(\mathcal{B}+1)$ are polynomially equivalent. Thus the classification for QCSPs may be considered a project at least as hard as that for CSPs. The following is the merger of Conjectures 6 and 7 in \cite{Meditations} which we call the \emph{Chen Conjecture}.
\begin{conjecture}[Chen Conjecture]
Let $\mathcal{B}$ be a finite relational structure expanded with all constants. If Pol$(\mathcal{B})$ has PGP, then QCSP$(\mathcal{B})$ is in NP; otherwise QCSP$(\mathcal{B})$ is Pspace-complete.
\end{conjecture}
In \cite{Meditations}, Conjecture 6 gives the NP membership and Conjecture 7 the Pspace-completeness. We now know from \cite{ZhukGap2015} and \cite{LICS2015} that the NP membership of Conjecture 6 is indeed true. The most interesting result of this paper is Theorem~\ref{thm:all} below, but note that we permit infinite signatures (languages) although our domains remain finite. This aspect of our work will be discussed in detail later. 
\begin{theorem}[Revised Chen Conjecture]
\label{thm:all}
Let $\mathbb{A}$ be an idempotent algebra on a finite domain $A$. If $\mathbb{A}$ satisfies PGP, then QCSP$(\mathrm{Inv}(\mathbb{A}))$ is in NP. Otherwise, QCSP$(\mathrm{Inv}(\mathbb{A}))$ is co-NP-hard.
\end{theorem}
Zhuk has previously proved \cite{ZhukGap2015} that only the cases PGP and EGP may occur, as well even in the non-idempotent case. With infinite languages, the NP-membership for Theorem~\ref{thm:all} is no longer immediate from \cite{LICS2015}, but requires a little extra work.  We are also able to refute the following form.
\begin{conjecture}[Alternative Chen Conjecture]
\label{thm:alternative}
Let $\mathbb{A}$ be an idempotent algebra on a finite domain $A$. If $\mathbb{A}$ satisfies PGP, then for every finite subset  $\Delta \subseteq \mathrm{Inv}(\mathbb{A})$, QCSP$(\Delta)$ is in NP. Otherwise, there exists a finite subset $\Delta \subseteq \mathrm{Inv}(\mathbb{A})$ so that QCSP$(\Delta)$ is co-NP-hard.
\end{conjecture}
In proving Theorem~\ref{thm:all} we are saying that the complexity of QCSPs, with all constants included, is classified modulo the complexity of CSPs. 
\begin{corollary}
Let $\mathbb{A}$ be an idempotent algebra on a finite domain $A$. Either QCSP$(\mathrm{Inv}(\mathbb{A}))$ is co-NP-hard or  QCSP$(\mathrm{Inv}(\mathbb{A}))$ has the same complexity as CSP$(\mathrm{Inv}(\mathbb{A}))$.
\end{corollary}
In this manner, our result follows in the footsteps of the similar result for the Valued CSP, which has also had its complexity classified modulo the CSP, as culminated in the paper \cite{FOCS2015}.

For a finite-domain algebra $\mathbb{A}$ we associate a function
$f_\mathbb{A}:\mathbb{N}\rightarrow\mathbb{N}$, giving the cardinality
of the minimal generating sets of the sequence $\mathbb{A},
\mathbb{A}^2, \mathbb{A}^3, \ldots$ as $f(1), f(2), f(3), \ldots$,
respectively. A subset $\Lambda$ of $A^m$ is a generating set for $\mathbb{A}^m$ exactly if, for every $(a_1,\ldots,a_m) \in A^m$, there exists a $k$-ary term operation $f$ of $\mathbb{A}$ and $(b^1_1,\ldots,b^1_m),\ldots,(b^k_1,\ldots,b^k_m) \in \Lambda$ so that $f(b^1_1,\ldots,b^k_1)=a_1$, \ldots, $f(b^1_m,\ldots,b^k_m)=a_m$. We may say $\mathbb{A}$ has the $g$-GP if $f(m) \leq
g(m)$ for all $m$. The question then arises as to the growth rate of
$f$ and specifically regarding the behaviours constant, logarithmic,
linear, polynomial and exponential. Wiegold proved in
\cite{WiegoldSemigroups} that if $\mathbb{A}$ is a finite semigroup
then $f_{\mathbb{A}}$ is either linear or exponential, with the former
prevailing precisely when $\mathbb{A}$ is a monoid. This dichotomy
classification may be seen as a gap theorem because no growth rates
intermediate between linear and exponential may occur. We say
$\mathbb{A}$  enjoys the \emph{polynomially generated powers} property
(PGP) if there exists a polynomial $p$ so that $f_{\mathbb{A}}=O(p)$
and  the \emph{exponentially generated powers} property (EGP) if there
exists a constant $b$ so that $f_{\mathbb{A}}=\Omega(g)$ where
$g(i)=b^i$.

In Hubie Chen's \cite{AU-Chen-PGP}, a new link between algebra and
QCSP was discovered. Chen's previous work in QCSP tractability largely
involved the special notion of \emph{Collapsibility}
\cite{hubie-sicomp}, but in \cite{AU-Chen-PGP} this was extended to a \emph{computationally effective} version of the PGP. For a finite-domain, idempotent algebra $\mathbb{A}$, \emph{$k$-collapsibility} may be seen as that special form of the PGP in which the generating set for $\mathbb{A}^m$ is constituted of all tuples $(x_1,\ldots,x_m)$ in which at least $m-k$ of these elements are equal. \emph{$k$-switchability} may be seen as another special form of the PGP in which the generating set for $\mathbb{A}^m$ is constituted of all tuples $(x_1,\ldots,x_m)$ in which there exists $a_i<\ldots<a_{k'}$, for $k'\leq k$, so that
\[ (x_1,\ldots,x_m) = (x_1,\ldots,x_{a_1},x_{a_1+1},\ldots,x_{a_2},x_{a_2+1},\ldots,\ldots,x_{a_k'},x_{a_k'+1},\ldots,x_m),\]
where $x_1=\ldots=x_{a_1-1}$, $x_{a_1}=\ldots=x_{a_2-1}$, \ldots, $x_{a_{k'}}=\ldots=x_{a_m}$. Thus, $a_1,a_2,\ldots,a_{k'}$ are the indices where the tuple switches value.  Note that these are not the original definitions, which we will see shortly, but they are proved equivalent to the original definitions (at least for finite signatures) in \cite{LICS2015}. Moreover, these are the definitions that we will use. We say that $\mathbb{A}$ is collapsible (switchable) if there exists $k$ such that it is $k$-collapsible ($k$-switchable). We note that Zhuk uses this definition of switchability in \cite{ZhukGap2015} in which he proved that the only kind of PGP for finite-domain algebras is switchability.

Let us capitalise Collapsibility and Switchability to indicate Chen's original definitions from \cite{AU-Chen-PGP} are used, following an example for \emph{arithmetic} versus \emph{Arithmetic} by Raymond Smullyan in \cite{smullyan1992godel}. There is the potential for confusion at the start of the sentence but, as was the case with Smullyan, the two will transpire to be interchangeable throughout our discourse. It is straightforward to see that $k$-Switchability implies $k$-switchability and $k$-Collapsibility implies $k$-collapsibility. The converses, for finite signatures, also hold, but this requires rather more work \cite{LICS2015}. For any finite algebra, $k$-Collapsibility implies $k$-Switchability. and for any $2$-element algebra, $k$-Switchability implies $k$-Collapsibility. Chen originally introduced Switchability because he found a $3$-element algebra that enjoyed the PGP but was not Collapsible \cite{AU-Chen-PGP}. He went on to prove that Switchability of $\mathbb{A}$ implies that the corresponding QCSP is in P, what one might informally state as QCSP$(\mathrm{Inv}(\mathbb{A}))$ in P, where $\mathrm{Inv}(\mathbb{A})$ can be seen as the structure over the same domain as $\mathbb{A}$ whose relations are precisely those that are preserved by (invariant under) all the operations of $\mathbb{A}$.
However, the QCSP was traditionally defined only on finite sets of relations (else the question arises as to encoding), thus a more formal definition might be that, for any finite subset $\Delta$ of $\mathrm{Inv}(\mathbb{A})$, QCSP$(\Delta)$ is in P. What we prove in this paper is that, as far as the QCSP is concerned, Switchability on a three-element algebra $\mathbb{A}$ is something of a mirage. What we mean by this is that when $\mathbb{A}$ is Switchable, for all finite subsets $\Delta$ of Inv$(\mathbb{A})$, already Pol$(\Delta)$ is Collapsible.
Thus, for QCSP complexity for three-element structures, we do not need the additional notion of Switchability to explain tractability, as Collapsibility will already suffice. Since these notions were originally introduced in connection with the QCSP this is particularly surprising. Note that the parameter $k$ of Collapsibility is unbounded over these increasing finite subsets $\Delta$ while the parameter of Switchability clearly remains bounded. In some way we are suggesting that Switchability itself might be seen as a limit phenomenon of Collapsibility.

\subsection{Infinite languages}

Our use of infinite languages (\mbox{i.e.} signatures, since we work on a finite domain) is the only controversial part of our discourse and merits special discussion. We wish to argue that a necessary corollary of the algebraic approach to (Q)CSP is a reconciliation with infinite languages. The traditional approach to consider arbitrary finite subsets of Inv$(\mathbb{A})$ is unsatisfactory in the sense that choosing this way to escape the -- naturally infinite -- set Inv$(\mathbb{A})$ is as arbitrary as the choice of encoding required for infinite languages. However, the difficulty in that choice is of course the reason why this route is often eschewed. The first possibility that comes to mind for encoding a relation in Inv$(\mathbb{A})$ is probably to list its tuples, while the second is likely to be to describe the relation in some kind of ``simple'' logic. Both these possibilities are discussed in \cite{Creignou}, for the Boolean domain, where the ``simple'' logic is the propositional calculus. For larger domains, this would be equivalent to quantifier-free propositions over equality with constants. Both Conjunctive Normal Form (CNF) and Disjunctive Normal Form (DNF) representations are considered in \cite{Creignou} and a similar discussion in \cite{ecsps} exposes the advantages of the DNF encoding. The point here is that testing non-emptiness of a relation encoded in CNF may already be NP-hard, while for DNF this will be tractable. Since DNF has some benign properties, we might consider it a ``nice, simple'' logic while for ``simple'' logic we encompass all quantifier-free sentences, that include DNF and CNF as special cases. The reason we describe this as ``simple'' logic is to compare against something stronger, say all first-order sentences over equality with constants. Here recognising non-emptiness becomes Pspace-hard and since QCSPs already sit in Pspace, this complexity is unreasonable.

For the QCSP over infinite languages Inv$(\mathbb{A})$, Chen and Mayr \cite{QCSPmonoids} have declared for our first, tuple-listing, encoding. In this paper we will choose the ``simple'' logic encoding, occasionally giving more refined results for its ``nice, simple'' restriction to DNF. Our choice of the ``simple'' logic encoding over the tuple-listing encoding will ultimately be justified by the (Revised) Chen Conjecture holding for ``simple'' logic yet failing for tuple-listings. Note that our demonstration of the (Revised) Chen Conjecture for infinite languages with the ``simple'' logic encoding does not resolve the original Chen Conjecture for finite languages $\mathcal{B}$ with constants because QCSP$(\mathrm{Inv}(\mathrm{Pol}(\mathcal{B})))$ could conceivably have higher complexity than QCSP$(\mathcal{B})$ due to a succinct representation of relations in $\mathrm{Inv}(\mathrm{Pol}(\mathcal{B}))$. Indeed, this belies one justification for the preferential study of finite subsets of $\mathrm{Inv}(\mathrm{Pol}(\mathcal{B}))$, since for finite signature $\mathcal{B}$ we can then say QCSP$(\mathcal{B})$ and QCSP$(\mathrm{Inv}(\mathrm{Pol}\mathcal{B}))$ must have the same complexity. Note that for finite relational bases $\mathcal{B}',\mathcal{B}''$ of $\mathrm{Inv}(\mathrm{Pol}(\mathcal{B}))$, QCSP$(\mathcal{B}')$ and QCSP$(\mathcal{B}'')$ must have the same complexity. Further, we do not know of any concrete finite $\mathcal{B}$ with constants, so that QCSP$(\mathrm{Inv}(\mathrm{Pol}(\mathcal{B})))$ and QCSP$(\mathcal{B})$ have different complexity. 

Let us consider examples of our encodings. For the domain $\{1,2,3\}$, we may give a binary relation either by the tuples $\{ (1,2), (2,1), (2,3), (3,2), (1,3), (3,1), (1,1) \}$ or by the ``simple'' logic formula $(x\neq y \vee x=1)$. For the domain $\{0,1\}$, we may give the ternary (not-all-equal) relation by the tuples $\{ (1,0,0), (0,1,0), (0,0,1), (1,1,0), (1,0,1), (1,1,0)\}$ or by the ``simple'' logic formula $(x\neq y \vee y \neq z)$. In both of these examples, the simple formula is also in DNF.

\vspace{0.2cm}
\noindent \textbf{Nota Bene}. The results of this paper apply for the ``simple'' logic encoding as well as the ``nice, simple'' encoding in DNF except where specifically stated otherwise. These exceptions are Proposition~\ref{prop:coNP2} and Corollary~\ref{cor:coNP2} (which  uses the ``nice, simple'' DNF) and Proposition~\ref{prop:Chen-fails} (which uses the tuple encoding).

\subsection{Related work}

This paper is the merger of \cite{MartinZhuk15,ChenConjectureQCSPs}, neither of which was submitted for publication, considerably extended.
 
\section{Preliminaries}

Let $[k]:=\{1,\ldots,k\}$. A \emph{$k$-ary polymorphism} of a relational structure $\mathcal{B}$ is a homomorphism $f$ from $\mathcal{B}^k$ to $\mathcal{B}$. Let Pol$(\mathcal{B})$ be the set of polymorphisms of $\mathcal{B}$ and let Inv$(\mathbb{A})$ be the set of relations on $A$ which are invariant under (each of) the operations of some finite algebra $\mathbb{A}$. Pol$(\mathcal{B})$ is an object known in Universal Algebra as a \emph{clone}, which is a set of operations containing all projections and closed under composition (superposition). A \emph{term operation} of an algebra $\mathbb{A}$ is an operation which is a member of the clone generated by $\mathbb{A}$.

We will conflate sets of operations over the same domain and algebras just as we do sets of relations over the same domain and constraint languages (relational structures). Indeed, the only technical difference between such objects is the movement away from an ordered signature, which is not something we will ever need. A \emph{reduct} of a relational structure $\mathcal{B}$ is a relational structure $\mathcal{B}'$ over the same domain obtained by forgetting some of the relations. If $\Delta$ is some finite subset of Inv$(\mathbb{A})$, then we may view $\Delta$ a being a finite reduct of the structure (associated with) Inv$(\mathbb{A})$.

A $k$-ary operation $f$ over $A$ is a \emph{projection} if $f(x_1,\ldots,x_k)=x_i$, for some $i \in [k]$.  When $\alpha,\beta$ are strict subsets of $A$ so that $\alpha \cup \beta=A$, then a $k$-ary operation $f$ on $A$ is said to be \emph{$\alpha\beta$-projective} if there exists $i \in [k]$ so that if $x_i \in \alpha$ (respectively, $x_i \in \beta$), then $f(x_1,\ldots,x_k) \in \alpha$ (respectively, $f(x_1,\ldots,x_k) \in \beta$).
 
We recall QCSP$(\mathcal{B})$, where $\mathcal{B}$ is some structure on a finite-domain, is a decision problem with input $\phi$, a pH-sentence (\mbox{i.e.} using just $\forall$, $\exists$, $\wedge$ and $=$) involving (a finite set of) relations of $\mathcal{B}$, encoded in propositional logic with equality and constants. The yes-instances are those $\phi$ for which $\mathcal{B} \models \phi$. If the input sentence is restricted to have alternation $\Pi_k$ then the corresponding problem is designated $\Pi_k$-CSP$(\mathcal{B})$.
 
\subsection{Games, adversaries and reactive composition}
\label{sec:games-adversaries-reactivecomposition}

We now recall some terminology due to
Chen~\cite{hubie-sicomp,AU-Chen-PGP}, for his natural adaptation of
the model checking game to the context of pH-sentences. 
We shall not need to explicitly play these games but only to handle
strategies for the existential player. This will enable us to give the original definitions for Collapsibility and Switchability.
An \emph{adversary} $\mathscr{B}$ of length $m\geq 1$ is an $m$-ary
relation over $A$. 
When $\mathscr{B}$ is precisely the set
$B_{1}\times B_{2} \times \ldots \times B_{m}$ for some non-empty
subsets $B_1,B_2,\ldots,B_m$ of $A$, we speak of a \emph{rectangular
  adversary} (we will sometimes specify this as a tuple rather than a product).
Let $\phi$ be a pH-sentence with universal variables $x_1,\ldots,x_m$ and quantifier-free
part $\psi$. We write $\mathcal{A}\models \phi_{\restrict\mathscr{B}}$ and say that
\emph{the existential player has a winning strategy in the
$(\mathcal{A},\phi)$-game against adversary $\mathscr{B}$} iff there
exists a set of Skolem functions $\{\sigma_x : \mbox{`$\exists x$'}
\in \phi \}$ such that for any assignment $\pi$ of the universally
quantified variables of $\phi$ to $A$, where
$\bigl(\pi(x_1),\ldots,\pi(x_m)\bigr) \in \mathscr{B}$, the map $h_\pi$ 
is a homomorphism from $\mathcal{D}_\psi$ (the canonical database) to $\mathcal{A}$, where
$$h_\pi(x):=
\begin{cases}
  \pi(x) & \text{, if $x$ is a universal variable; and,}\\
  \sigma_x(\left.\pi\right|_{Y_x})& \text{, otherwise.}\\
\end{cases}
$$
(Here, $Y_x$ denotes the set of universal variables preceding $x$ and $\left.\pi\right|_{Y_x}$ the restriction of $\pi$ to $Y_x$.)
Clearly, $\mathcal{A} \models \phi$ iff the existential player has a winning strategy in the
$(\mathcal{A},\phi)$-game against the so-called \emph{full
  (rectangular) 
  adversary} $A\times A \times \ldots \times A$ (which
we will denote hereafter by $A^m$).
We say that an adversary $\mathscr{B}$ of length $m$ \emph{dominates}
an adversary $\mathscr{B}'$ of length $m$ when $\mathscr{B}'\subseteq
\mathscr{B}$. Note that $\mathscr{B}'\subseteq  \mathscr{B}$ and
$\mathcal{A}\models \phi_{\restrict\mathscr{B}}$ implies
$\mathcal{A}\models \phi_{\restrict\mathscr{B}'}$. 
We will also
consider sets of adversaries of the same length, denoted
by uppercase Greek letters as in $\Omega_m$ (here the length is $m$); and, sequences thereof, which we
denote with bold uppercase Greek letters as in
$\mathbf{\Omega}=\bigl(\Omega_m\bigr)_{m \in \integerset}$. 
We will write
$\mathcal{A}\models \phi_{\restrict\Omega_m}$ to denote that
$\mathcal{A}\models  \phi_{\restrict\mathscr{B}}$ holds for every
adversary $\mathscr{B}$ in $\Omega_m$. 
%

Let $f$ be a $k$-ary operation of $\mathcal{A}$ and $\mathscr{A},\mathscr{B}_1,\ldots,\mathscr{B}_k$ be adversaries of length $m$.
We say that $\mathscr{A}$ is \emph{reactively composable} from the
adversaries $\mathscr{B}_1,\ldots,\mathscr{B}_k$ via $f$, and we write $\mathscr{A} \reactivelycomposable f(\mathscr{B}_1,\ldots,\mathscr{B}_k)$ iff there exist partial functions $g^j_i:A^i \to A$ for every $i$ in $[m]$ and every $j$ in $[k]$ such that, for every tuple $(a_1,\ldots,a_m)$ in adversary $\mathscr{A}$ the following holds.
\begin{itemize}
\item for every $j$ in $[k]$, the values 
  $g^j_1(a_1), g^j_2(a_1,a_2),$ $\ldots , g^j_m(a_1,a_2,\ldots,a_m)$ are defined and the tuple $\bigl(g^j_1(a_1), g^j_2(a_1,a_2), \ldots, g^j_m(a_1,a_2,\ldots,a_m)\bigr)$ is in adversary $\mathscr{B}_j$; and,
\item for every $i$ in $[m]$, $a_i =f\bigl(g^1_i(a_1,a_2,\ldots,a_i),$
  $g^2_i(a_1,a_2,\ldots,a_i),\ldots,g^k_i(a_1,a_2,\ldots,a_i))$.
\end{itemize}
We write $\mathscr{A} \reactivelycomposable
\{\mathscr{B}_1,\ldots,\mathscr{B}_k\}$ if there exists a $k$-ary
operation $f$ such that $\mathscr{A} \reactivelycomposable f(\mathscr{B}_1,\ldots,\mathscr{B}_k)$
Reactive composition allows to interpolate complete Skolem functions
from partial ones.
\begin{theorem}[{\cite[Theorem 7.6]{AU-Chen-PGP}}]
  \label{thm:hubieReactiveComposition}
  Let $\phi$ be a pH-sentence with $m$ universal variables. Let $\mathscr{A}$ be an adversary and $\Omega_m$ a set of adversaries, both of length $m$.
  
  If $\mathcal{A}\models \phi_{\restrict\Omega_m}$ and  $\mathscr{A} \reactivelycomposable \Omega_m$ then
  $\mathcal{A} \models \phi_{\restrict\mathscr{A}}$.
\end{theorem}
As a concrete example of an interesting sequence of adversaries, consider the adversaries for the notion of
\emph{$p$-Collapsibility}. Let $p\geq 0$ be some fixed integer.
For $x$ in $A$, let $\Upsilon_{m,p,x}$ be the set of all 
rectangular 
adversaries of length $m$ with $p$ co-ordinates that are the set $A$
and all the others that are the fixed singleton $\{x\}$. For
$B\subseteq A$, let $\Upsilon_{m,p,B}$ be the union of $\Upsilon_{m,p,x}$ for all $x$ in $B$.
Let $\mathbf\Upsilon_{p,B}$ be the sequence of adversaries
$\Bigl(\Upsilon_{m,p,B} \Bigr)_{m \in \integerset}$.
We will define a structure $\mathcal{A}$ to be \emph{$p$-Collapsible
  from source $B$} iff for every $m$ and for all pH-sentence $\phi$
with $m$ universal variables, $\mathcal{A}\models \phi_{\restrict\Upsilon_{m,p,B}}$ implies   $\mathcal{A} \models \phi$. 

For $p$-Switchability, the adversaries will be of the form $\Xi_{m,p}$ which contains all tuples which have no more than $p$ switches.

For rectangular adversaries, such as $\Upsilon_{m,p,x}$, reactive composition is rather simpler than in the definition above, becoming just (ordinary) composition, as follows. $\mathscr{A}$ is \emph{composable} from the adversaries $\mathscr{B}_1,\ldots,\mathscr{B}_k$ via $f$ if $f(B^i_1,\ldots,B_i^k) \supseteq A^i$, where $\mathscr{A}=(A^1, \ldots, A^m)$ and each $\mathscr{B}_j=(B^1_j,\ldots,B^m_j)$. Reactive composition plays a key role in the proof of our main theorem but its use appears only in other papers that we will cite. Ordinary composition is the only type of reactive composition that will be used in this paper.

\section{The Chen Conjecture}

\subsection{NP-membership}

We need to revisit the main result of \cite{LICS2015} to show that it holds not just for finite signatures but for infinite signatures also. In its original the following theorem discussed ``projective sequences of adversaries, none of which are degenerate''. This includes Switching adversaries and we give it in this latter form. We furthermore remove some parts of the theorem that are not currently relevant to us.
  \begin{theorem}[\textbf{In abstracto} \cite{LICS2015}]\label{MainResult:InAbstracto}
    Let $\mathbf{\Omega}=\bigl(\Omega_m\bigr)_{m \in \integerset}$ be the sequence of the set of all ($k$-)Switching $m$-ary adversaries over the domain of $\mathcal{A}$, a finite structure. The following are equivalent.
  \begin{romannum}
  \item[$(i)$] For every $m\geq 1$, for every pH-sentence $\psi$ with $m$ universal
    variables, $\mathcal{A}\models \psi_{\restrict \Omega_{m}}$
    implies $\mathcal{A}\models \psi$.
  \item[$(vi)$] For every $m\geq 1$, $\Omega_m$ generates $\mathrm{Pol}(\mathcal{A})^m$.
  \end{romannum}
\end{theorem}
 \begin{corollary}[\textbf{In abstracto levavi}]\label{MainResult:InAbstractoLevavi}
    Let $\mathbf{\Omega}=\bigl(\Omega_m\bigr)_{m \in \integerset}$ be the sequence of the set of all ($k$-)Switching $m$-ary adversaries over the domain of $\mathcal{A}$, a finite-domain structure with an infinite signature. The following are equivalent.
  \begin{romannum}
  \item[$(i)$] For every $m\geq 1$, for every pH-sentence $\psi$ with $m$ universal
    variables, $\mathcal{A}\models \psi_{\restrict \Omega_{m}}$
    implies $\mathcal{A}\models \psi$.
  \item[$(vi)$] For every $m\geq 1$, $\Omega_m$ generates $\mathrm{Pol}(\mathcal{A})^m$.
  \end{romannum}
\end{corollary}

\begin{proof}
We know from Theorem~\ref{MainResult:InAbstracto} that the following are equivalent:
 \begin{romannum}
  \item[$(i')$] For every finite-signature reduct $\mathcal{A}'$ of $\mathcal{A}$ and $m\geq 1$, for every pH-sentence $\psi$ with $m$ universal
    variables, $\mathcal{A}' \models \psi_{\restrict \Omega_{m}}$
    implies $\mathcal{A}' \models \psi$.
  \item[$(vi')$] For every finite-signature reduct $\mathcal{A}'$ of $\mathcal{A}$ and every $m\geq 1$, $\Omega_m$ generates $\mathrm{Pol}(\mathcal{A}')^m$.
  \end{romannum}
Since it is clear that both  $(i) \Rightarrow (i')$ and $(vi) \Rightarrow (vi')$, it remains to argue that $(i') \Rightarrow (i)$ and $(vi') \Rightarrow (vi)$.

[$(i') \Rightarrow (i)$.] By contraposition, if $(i)$ fails then it fails on some specific  pH-sentence $\psi$ which only mentions a finite number of relations of $\mathcal{A}'$. Thus $(i')$ also fails on some finite reduct of $\mathcal{A}'$ mentioning these relations.

[$(vi') \Rightarrow (vi)$.] Let $m$ be given. Consider some chain of finite reducts $\mathcal{A}_1,\ldots,\mathcal{A}_2,\ldots$ of $\mathcal{A}$ so that each $\mathcal{A}_i$ is a reduct of $\mathcal{A}_j$ for $i<j$ and every relation of $\mathcal{A}$ appears in some $\mathcal{A}_i$. We can assume from $(vi)'$ that $\Omega_m$ generates $\mathrm{Pol}(\mathcal{A}_i)^m$, for each $i$. But since the number of tuples $(a_1,\ldots,a_m)$ and operations from $\Omega_m$ to $(a_1,\ldots,a_m)$ witnessing generation in $\mathrm{Pol}(\mathcal{A}')^m$ is finite, the sequence of operations $(f^i_1,\ldots,f^i_{|A|^m})$ witnessing these must have an infinitely recurring element as $i$ tends to infinity. One such recurring element we call $(f_1,\ldots,f_{|A|^m})$ and this witnesses generation in $\mathrm{Pol}(\mathcal{A})^m$.
\end{proof}
Note that in $(vi') \Rightarrow (vi)$ above we did not need to argue uniformly across the different $(a_1,\ldots,a_m)$ and it is enough to find an infinitely recurring operation for each of these individually.

The following result is essentially a corollary of the works of Chen and Zhuk \cite{AU-Chen-PGP,ZhukGap2015} via \cite{LICS2015}.
\begin{theorem}
\label{thm:easy}
Let $\mathbb{A}$ be an idempotent algebra on a finite domain $A$. If $\mathbb{A}$ satisfies PGP, then QCSP$(\mathrm{Inv}(\mathbb{A}))$ reduces to a polynomial number of instances of CSP$(\mathrm{Inv}(\mathbb{A}))$ and is in NP.
\end{theorem}
\begin{proof}
We know from Theorem 7 in \cite{ZhukGap2015} that $\mathbb{A}$ is Switchable, whereupon we apply Corollary~\ref{MainResult:InAbstractoLevavi}, $(vi) \Rightarrow (i)$. By considering instances whose universal variables involve only the polynomial number of tuples from the Switching Adversary, one can see that QCSP$(\mathrm{Inv}(\mathbb{A}))$ reduces to a polynomial number of instances of CSP$(\mathrm{Inv}(\mathbb{A}))$ and is therefore in NP. Further details of the NP algorithm are given in Corollary 38 of \cite{LICS2015} but the argument here follows exactly Section 7 from \cite{AU-Chen-PGP}, in which it was originally proved that Switchability yields the corresponding QCSP in NP.
\end{proof}

Note that Chen's original definition of Switchability, based on adversaries and reactive composability, plays a key role in the NP membership algorithm in Theorem~\ref{thm:easy}. It is the result from \cite{LICS2015} that is required to reconcile the two definitions of switchability as equivalent, and indeed Corollary~\ref{MainResult:InAbstractoLevavi} is needed in this process for infinite signatures. If we were to use just our definition of switchability then it is only possible to prove, \`a la Proposition 3.3 in \cite{AU-Chen-PGP}, that the bounded alternation $\Pi_n$-CSP$(\mathrm{Inv}(\mathbb{A}))$ is in NP. Thus, using just the methods from \cite{AU-Chen-PGP} and \cite{ZhukGap2015}, we can not prove the Revised Chen Conjecture, but rather some bounded alternation (re)revision.

\subsection{co-NP-hardness}

Suppose there exist $\alpha,\beta$ strict subsets of $A$ so that $\alpha \cup \beta = A$, define the relation $\tau_k(x_1,y_1,z_1\ldots,x_k,y_k,z_k)$ defined by
\[ \tau_k(x_1,y_1,z_1\ldots,x_k,y_k,z_k):=\rho'(x_1,y_1,z_1) \vee \ldots \vee \rho'(x_k,y_k,z_k),\]
where $\rho'(x,y,z)=(\alpha \times \alpha \times \alpha) \cup (\beta \times \beta \times \beta)$. Strictly speaking, the $\alpha$ and $\beta$ are parameters of $\tau_k$ but we dispense with adding them to the notation since they will be fixed at any point in which we invoke the $\tau_k$. The purpose of the relations $\tau_k$ is to encode co-NP-hardness through the complement of the problem  (monotone) \emph{$3$-not-all-equal-satisfiability} (3NAESAT). Let us introduce also the important relations $\sigma_k(x_1,y_1,\ldots,x_k,y_k)$ defined by
\[ \sigma_k(x_1,y_1,\ldots,x_k,y_k):=\rho(x_1,y_1) \vee \ldots \vee \rho(x_k,y_k),\]
where $\rho(x,y)=(\alpha \times \alpha) \cup (\beta \times \beta)$.
\begin{lemma}
The relation $\tau_k$ is pp-definable in $\sigma_k$.
\label{lem:new-revision}
\end{lemma}
\begin{proof}
We will argue that $\tau_k$ is definable by the conjunction $\Phi$ of $3^k$ instances of $\sigma_k$ that each consider the ways in which two variables may be chosen from each of the $(x_i,y_i,z_i)$, i.e. $x_i\sim y_i$ or $y_i\sim z_i$ or $x_i\sim z_i$ (where $\sim$ is infix for $\rho$). We need to show that this conjunction $\Phi$ entails $\tau_k$ (the converse is trivial). We will assume for contradiction that $\Phi$ is satisfiable but $\tau_k$ not. In the first instance of $\sigma_k$ of $\Phi$ some atom must be true, and it will be of the form $x_i\sim y_i$ or $y_i\sim z_i$ or $x_i\sim z_i$. Once we have settled on one of these three, $p_i\sim q_i$, then we immediately satisfy $3^{k-1}$ of the conjunctions of $\Phi$, leaving $2\cdot 3^{k-1}$ unsatisfied. Now, we can not evaluate true any of the others among $\{x_i\sim y_i, y_i\sim z_i, x_i\sim z_i\} \setminus \{p_i\sim q_i\}$ without contradicting our assumption. Thus we are now down to looking at variables with subscript other than $i$ and in this fashion we have made the space one smaller, in total $k-1$. Now, we will need to evaluate in $\Phi$ some other atom of the form  $x_j\sim y_j$ or $y_j\sim z_j$ or $x_j\sim z_j$, for $j\neq i$. Once we have settled on one of these three then we immediately satisfy $2 \cdot 3^{k-2}$ of the conjunctions remaining of $\Phi$, leaving $2^2 \cdot 3^{k-2}$ still unsatisfied. Iterating this thinking, we arrive at a situation in which $2^k$ clauses are unsatisfied after we have gone through all $k$ subscripts, which is a contradiction. 
\end{proof}
\begin{theorem}
\label{thm:hard}
Let $\mathbb{A}$ be an idempotent algebra on a finite domain $A$. If $\mathbb{A}$ satisfies EGP, then QCSP$(\mathrm{Inv}(\mathbb{A}))$ is co-NP-hard.
\end{theorem}
\begin{proof}
We know from Lemma 11 in \cite{ZhukGap2015} that there exist $\alpha,\beta$ strict subsets of $A$ so that $\alpha \cup \beta = A$ and the relation $\sigma_k$ is in $\mathrm{Inv}(\mathbb{A})$, for each $k \in \mathbb{N}$. From Lemma~\ref{lem:new-revision}, we know also that $\tau_k$  is in $\mathrm{Inv}(\mathbb{A})$, for each $k \in \mathbb{N}$.

We will next argue that $\tau_k$ enjoys a relatively small specification in DNF (at least, polynomial in $k$). We first give such a specification for $\rho'(x,y,z)$.
\[ \rho'(x,y,z):= \bigvee_{a,a',a'' \in \alpha} x=a \wedge y=a' \wedge z=a'' \vee \bigvee_{b,b',b'' \in \beta} x=b \wedge y=b' \wedge z=b''\]
which is constant in size when $A$ is fixed. Now it is clear from the definition that the size of $\tau_n$ is polynomial in $n$.

We will now give a very simple reduction from the complement of 3NAESAT to QCSP$(\mathrm{Inv}(\mathbb{A}))$. 3NAESAT is well-known to be NP-complete \cite{Papa} and our result will follow.

Take an instance $\phi$ of 3NAESAT which is the existential quantification of a conjunction of $k$ atoms $\mathrm{NAE}(x,y,z)$. Thus $\neg \phi$ is the universal quantification of a disjunction of $k$ atoms $x=y=z$. We build our instance $\psi$ of QCSP$(\mathrm{Inv}(\mathbb{A}))$ from $\neg \phi$ by transforming the quantifier-free part $x_1=y_1=z_1 \vee \ldots \vee x_k=y_k=z_k$ to $\tau_k=\rho'(x_1,y_1,z_1) \vee \ldots \vee \rho'(x_k,y_k,z_k)$.

($\neg \phi \in \mathrm{co\mbox{-}3NAESAT}$ implies $\psi \in \mathrm{QCSP}(\mathrm{Inv}(\mathbb{A}))$.) From an assignment to the universal variables $v_1,\ldots,v_m$ of $\psi$ to elements $x_1,\ldots,x_m$ of $A$, consider elements $x'_1,\ldots,x'_m \in \{0,1\}$ according to 
\begin{itemize}
\item $x_i \in \alpha \setminus \beta$ implies $x'_i=0$, 
\item $x_i \in \beta \setminus \alpha$ implies $x'_i=1$, and
\item $x_i \in \alpha \cap \beta$ implies we don't care, so w.l.o.g. say $x'_i=0$.
\end{itemize}
The disjunct that is satisfied in the quantifier-free part of $\neg \phi$ now gives the corresponding disjunct that will be satisfied in $\tau_k$.

($\psi \in \mathrm{QCSP}(\mathrm{Inv}(\mathbb{A}))$ implies $\neg \phi \in \mathrm{co\mbox{-}3NAESAT}$.) From an assignment to the universal variables $v_1,\ldots,v_m$ of $\neg \phi$ to elements $x_1,\ldots,x_m$ of $\{0,1\}$, consider elements $x'_1,\ldots,x'_m \in A$ according to 
\begin{itemize}
\item $x_i=0$ implies $x'_i$ is some arbitrarily chosen element in $\alpha \setminus \beta$, and
\item $x_i=1$ implies $x'_i$ is some arbitrarily chosen element in $\beta \setminus \alpha$.
\end{itemize}
The disjunct that is satisfied in $\tau_k$ now gives the corresponding disjunct that will be satisfied in the quantifier-free part of $\neg \phi$.
\end{proof}
\noindent The demonstration of co-NP-hardness in the previous theorem was inspired by a similar proof in \cite{BodirskyChenSICOMP}. Note that an alternative proof that $\tau_k$ is in $\mathrm{Inv}(\mathbb{A})$ is furnished by the observation that it is preserved by all $\alpha\beta$-projections (see \cite{ZhukGap2015}). We note surprisingly that co-NP-hardness in Theorem~\ref{thm:hard} is optimal, in the sense that some (but not all!) of the cases just proced co-NP-hard are also in co-NP.
\begin{proposition}
Let $\alpha,\beta$ strict subsets of $A:=\{a_1,\ldots,a_n\}$ so that $\alpha \cup \beta = A$ and $\alpha \cap \beta \neq \emptyset$. Then QCSP$(A;\{\tau_k:k \in \mathbb{N}\},a_1,\ldots,a_n)$ is in co-NP.
\label{prop:coNP1}
\end{proposition}
\begin{proof}
Assume $|A|>1$, \mbox{i.e.} $n>1$ (note that the proof is trivial otherwise).
Let $\phi$ be an input to QCSP$(A;\{\tau_k:k \in \mathbb{N}\},a_1,\ldots,a_n)$. We will now seek to eliminate atoms $v=a$ ($a \in \{a_1,\ldots,a_n\}$) from $\phi$. Suppose $\phi$ has an atom $v=a$. If $v$ is universally quantified, then $\phi$ is false (since $|A|>1$). Otherwise, either the atom $v=a$ may be eliminated with the variable $v$ since $v$ does not appear in a non-equality relation; or $\phi$ is false because there is another atom $v=a'$ for $a\neq a'$; or $v=a$ may be removed by substitution of $a$ into all non-equality instances of relations involving $v$. This preprocessing procedure is polynomial and we will assume \mbox{w.l.o.g.} that $\phi$ contains no atoms $v=a$. We now argue that $\phi$ is a yes-instance iff $\phi'$ is a yes-instance, where $\phi'$ is built from $\phi$ by instantiating all existentially quantified variables as any $a \in \alpha \cap \beta$. The universal $\phi'$ can be evaluated in co-NP (one may prefer to imagine the complement as an existential $\neg \phi' $ to be evaluated in NP) and the result follows.
\end{proof}
In fact, this being an algebraic paper, we can even do better. Let $\mathcal{B}$ signify a set of relations on a finite domain but not necessarily itself finite. For convenience, we will assume the set of relations of $\mathcal{B}$ is closed under all co-ordinate projections and instantiations of constants. 
Call $\mathcal{B}$ \emph{existentially trivial} if there exists an element $c \in B$ (which we call a \emph{canon}) such that for each $k$-ary relation $R$ of $\mathcal{B}$ and each $i \in [k]$, and for every $x_1,\ldots,x_k \in B$, whenever $(x_1,\ldots,x_{i-1},x_i,x_{i+1},\ldots,x_k) \in R^{\mathcal{B}}$ then also $(x_1,\ldots,x_{i-1},c,x_{i+1},\ldots,x_k) \in R^{\mathcal{B}}$. We want to expand this class to \emph{almost existentially trivial} by permitting conjunctions of the form $v=a_i$ or $v=v'$ with relations that are existentially trivial. 
\begin{lemma}
Let $\alpha,\beta$ be strict subsets of $A:=\{a_1,\ldots,a_n\}$ so that $\alpha \cup \beta = A$ and $\alpha \cap \beta \neq \emptyset$. The set of relations pp-definable in $(A;\{\tau_k:k \in \mathbb{N}\},a_1,\ldots,a_n)$ is almost existentially trivial.
\end{lemma}
\begin{proof}
Consider a formula with a pp-definition in $(A;\{\tau_k:k \in \mathbb{N}\},a_1,\ldots,a_n)$. We assume that only free variables appear in equalities since otherwise we can remove these equalities by substitution. Now existential quantifiers can be removed and their variables instantiated as the canon $c$. Indeed, their atoms $\tau_n$ may now be removed since they will always be satisfied. Thus we are left with a conjunction of equalities and atoms $\tau_n$, and the result follows.
\end{proof}
\begin{proposition}
If $\mathcal{B}$ is comprised exclusively of relations that are almost existentially trivial, then QCSP$(\mathcal{B})$ is in co-NP under the \textbf{DNF encoding}.
\label{prop:coNP2}
\end{proposition}
\begin{proof}
The argument here is quite similar to that of Proposition~\ref{prop:coNP1} except that there is some additional preprocessing to find out variables that are forced in some relation to being a single constant or pairs of variables within a relation that are forced to be equal. In the first instance that some variable is forced to be constant in a $k$-ary relation, we should replace with the $(k-1)$-ary relation with the requisite forcing. In the second instance that a pair of variables are forced equal then we replace again the $k$-ary relation with a $(k-1)$-ary relation as well as an equality. Note that projecting a relation to a single or two co-ordinates can be done in polynomial time because the relations are encoded in DNF. After following these rules to their conclusion one obtains a conjunction of equalities together with relations that are existentially trivial. Now is the time to propagate variables to remove equalities (or find that there is no solution). Finally, when only existentially trivial relations are left, all remaining existential variables may be evaluated to the canon $c$.
\end{proof}
\begin{corollary}
Let $\alpha,\beta$ be strict subsets of $A:=\{a_1,\ldots,a_n\}$ so that $\alpha \cup \beta = A$ and $\alpha \cap \beta \neq \emptyset$. Then QCSP$(\mathrm{Inv}(\mathrm{Pol}(A;\{\tau_k:k \in \mathbb{N}\},a_,\ldots,a_n)))$ is in co-NP under the \textbf{DNF encoding}.
\label{cor:coNP2}
\end{corollary}
This last result, together with its supporting proposition, is the only time we seem to require the ``nice, simple'' DNF encoding, rather than arbitrary propositional logic. We do not require DNF for Proposition~\ref{prop:coNP1} as we have just a single relation in the signature for each arity and this is easy to keep track of. We note that the set of relations $\{\tau_k:k \in \mathbb{N}\}$ is not maximal with the property that with the constants it forms a co-clone of existentially trivial relations. One may add, for example, $\alpha \times \beta \cup \beta \times \alpha$.

The following, together with our previous results, gives the refutation of the Alternative Chen Conjecture.
\begin{proposition}
Let $\alpha,\beta$ strict subsets of $A:=\{a_1,\ldots,a_n\}$ so that $\alpha \cup \beta = A$ and $\alpha \cap \beta \neq \emptyset$. Then, for each finite signature reduct $\mathcal{B}$ of $(A;\{\tau_k:k \in \mathbb{N}\},a_1,\ldots,a_n)$, QCSP$(\mathcal{B})$ is in NL.
\label{prop:finite-NL}
\end{proposition}
\begin{proof}
We will assume $\mathcal{B}$ contains all constants (since we prove this case gives a QCSP in NL, it naturally follows that the same holds without constants). Take $m$ so that, for each $\tau_i \in \mathcal{B}$, $i\leq m$. Recall from Lemma~\ref{lem:new-revision} that $\tau_i$ is pp-definable in $\sigma_i$. We will prove that the structure $\mathcal{B}'$ given by $(A;\{\sigma_k:k \leq m \},a_1,\ldots,a_n)$ admits a $(3m+1)$-ary near-unanimity operation $f$ as a polymorphism, whereupon it follows that $\mathcal{B}$ admits the same near-unanimity polymorphism. We choose $f$ so that all tuples whose map is not automatically defined by the near-unanimity criterion map to some arbitrary $a \in \alpha \cap \beta$. To see this, imagine that this $f$ were not a polymorphism. Then some $(3m+1)$ tuples in $\sigma_i$ would be mapped to some tuple not in $\sigma_i$ which must be a tuple $\overline{t}$ of elements from $\alpha \setminus \beta \cup \beta \setminus \alpha$. Note that column-wise this map may only come from $(3m+1)$-tuples that have $3m$ instances of the same element. By the pigeonhole principle, the tuple $\overline{t}$ must appear as one of the $(3m+1)$ tuples in $\sigma_i$ and this is clearly a contradiction.

It follows from \cite{hubie-sicomp} that QCSP$(\mathcal{B})$ reduces to a polynomially bounded ensemble of ${n \choose 3m} \cdot n \cdot n^{3m}$ instances CSP$(\mathcal{B})$, and the result follows.
\end{proof}

\subsection{The question of the tuple encoding}

\begin{proposition}
Let $\alpha:=\{0,1\}$ and $\beta:=\{0,2\}$. Then, QCSP$(\{0,1,2\};\{\tau_k:k \in \mathbb{N}\},0,1,2)$ is in P under the \textbf{tuple encoding}.
\label{prop:Chen-fails}
\end{proposition}
\begin{proof}
Consider an instance $\phi$ of this QCSP of size $n$ involving relation $\tau_m$ but no relation $\tau_k$ for $k>m$. The number of tuples in $\tau_m$ is $>3^m$. 
Following Proposition~\ref{prop:coNP1} together with its proof, we may assume that the instance is strictly universally quantified over a conjunction of atoms (involving also constants). Now, a universally quantified conjunction is true iff the conjunction of its universally quantified atoms is true. We can further say that there are at most $n$ atoms each of which involves at most $3m$ variables. Therefore there is an exhaustive algorithm that takes at most $O(n\cdot 3^{3m})$ steps with is $O(n^4)$.
\end{proof}
The proof of Proposition~\ref{prop:Chen-fails} suggests an alternative proof of Proposition~\ref{prop:finite-NL}, but placing the corresponding QCSP in P instead of NL.
Proposition~\ref{prop:Chen-fails} shows that Chen's Conjecture fails for the tuple encoding in the sense that it provides a language $\mathcal{B}$, expanded with constants, so that Pol$(\mathcal{B})$ has EGP, yet QCSP$(\mathcal{B})$ is in P under the tuple encoding. However, it does not imply that the algebraic approach to QCSP violates Chen's Conjecture under the tuple encoding. This is because $(\{0,1,2\};\{\tau_k:k \in \mathbb{N}\},0,1,2)$ is not of the form Inv$(\mathbb{A})$ for some idempotent algebra $\mathbb{A}$. For this stronger result, we would need to prove QCSP$(\mathrm{Inv}(\mathrm{Pol}(\{0,1,2\};\{\tau_k:k \in \mathbb{N}\},0,1,2)))$ is in P under the tuple encoding.

\section{Switchability, Collapsability and the three-element case} 

An algebra $\mathbb{A}$ is a \emph{G-set} if its domain is not one-element and every of its operation $f$ is of the form $f(x_1, \ldots , x_k) = \pi(x_i)$ where $i \in [k]$ and $\pi$ is a permutation on A. An algebra $\mathbb{A}$ contains a G-set as a \emph{factor} if some homomorphic image of a subalgebra of $\mathbb{A}$ is a G-set. A \emph{Gap Algebra} \cite{hubie-sicomp} is a three-element idempotent algebra that omits a G-set as a factor and is not Collapsible.

Our first task is the deduction of the following theorem, whose lengthy proof appears in Appendix A. For each of the following two theorems, $\alpha$ and $\beta$ are chosen such that $\alpha, \beta$ are strict subsets of $\{0,1,2\}$, $\alpha \cup \beta = \{0,1,2\}$ and $\alpha \cap \beta \neq \emptyset$.

\begin{theorem}
\label{cor:Dmitriy-long}
Suppose $\mathbb{A}$ is a Gap Algebra that is not $\alpha\beta$-projective. Then, for every finite subset of $\Delta$ of Inv$(\mathbb{A})$, Pol$(\Delta)$ is Collapsible.
\end{theorem}

Our second task is the deduction of the following theorem, whose lengthy proof appears in Appendix B.

\begin{theorem}
\label{thm:catarina}
Suppose $\mathbb{A}$ is a $3$-element idempotent algebra that is not $\alpha\beta$-projective, containing a $2$-element G-set as a subalgebra. Then, $\mathbb{A}$ is Collapsible.
\end{theorem}

\begin{corollary}
\label{cor:3-element-collapsibility}
Suppose $\mathbb{A}$ is a $3$-element idempotent algebra that is not EGP, \mbox{i.e.} is Switchable. Then, for every finite subset of $\Delta$ of Inv$(\mathbb{A})$, Pol$(\Delta)$ is Collapsible.
\end{corollary}
\begin{proof}
Recall Lemma 11 in \cite{ZhukGap2015} that $\mathbb{A}$ has EGP iff there exists $\alpha$ and $\beta$ such that $\alpha, \beta$ are strict subsets of $D$, $\alpha \cup \beta = D$, and all operations of $\mathbb{A}$ are $\alpha\beta$-projective.
 
If $\mathbb{A}$ does not contain a G-set as a factor, then $\mathbb{A}$ is a Gap Algebra and the result follows from Theorem~\ref{cor:Dmitriy-long}. Otherwise, $\mathbb{A}$ contains a G-set as a factor. If $\mathbb{A}$ contains a G-set as a homomorphic image then $\mathbb{A}$ has EGP from \cite{AU-Chen-PGP}. Else, since $\mathbb{A}$ is $3$-element, $\mathbb{A}$ contains a $2$-element G-set as a subalgebra and we are in the situation of Theorem~\ref{thm:catarina}.
\end{proof}

\section{A three-element vignette}

We would love to be able to improve Theorem~\ref{thm:all} to describe the boundary between those cases that are co-NP-complete and those that are Pspace-complete, if indeed such a result is true.
However, even in the three-element case this appears challenging, but we are able to provide a variant vignette, whose proof appears in Appendix C.
\begin{theorem}
Let $\mathbb{A}$ be an idempotent algebra on a $3$-element domain. Either 
\begin{itemize}
\item $\Pi_k$-CSP$(\mathrm{Inv}(\mathbb{A}))$ is in NP, for all $k$; or
\item $\Pi_k$-CSP$(\mathrm{Inv}(\mathbb{A}))$ is co-NP-complete, for all $k$; or
\item $\Pi_k$-CSP$(\mathrm{Inv}(\mathbb{A}))$ is $\Pi^{\mathrm{P}}_2$-hard, for some $k$.
\end{itemize}
\label{thm:vignette}
\end{theorem}


Note that the trichotomy of Theorem~\ref{thm:vignette} does not hold for QCSP along the same boundary for, respectively, NP, co-NP-complete and Pspace-complete. For the semilattice-without-unit $s$ it is known that $\Pi_k$-CSP$(\mathrm{Inv}(s))$ is co-NP-complete, for all $k$, while  QCSP$(\mathrm{Inv}(s))$ is Pspace-complete \cite{BBCJK}.


\section{Discussion}

The major contribution of this paper is its discussion of the Chen Conjecture with two infinite-signature variants one of which is proved to hold (with encoding in ``simple logic'') and one of which fails (with the tuple listing).

In addition to this, the contribution is largely mathematical, examining the relationship between Switchability and Collapsibility in the three-element case. However, this mathematical study uncovers something of importance to the computer scientist who is not reconciled to infinite signatures! Since here it demonstrates that all three-element domain NP-memberships that may be shown by Switchability, may already be shown by Collapsibility. 

The work associated with Theorem~\ref{cor:Dmitriy-long} is distinctly non-trivial and involves a new method, whereas the work associated with Theorem~\ref{thm:catarina} uses known methods and involves mostly turning the handle with these. Similarly, the work involved with the three element vignette uses known methods on top of our earlier new results.


The Chen Conjecture in its original form remains open. As does the general question (for arbitrary finite domains) as to whether, if $\mathbb{A}$ is Switchable, all finite subsets $\mathcal{B}$ of Inv$(\mathbb{A})$ are so that Pol$(\mathcal{B})$ is Collapsible. However, to now prove the Chen Conjecture it is sufficient to prove, for any finite $\mathcal{B}$ expanded with all constants such that Pol$(\mathcal{B})$ has EGP, that there exists polynomially (in $i$) computable pp-definitions (over $\mathcal{B}$) of the relations $\tau_i$ (where $\alpha$ and $\beta$ are suitably chosen to witness EGP). A first step towards this is to establish whether there are even polynomially sized pp-definitions of these $\tau_i$.

The appearance of a co-NP-complete QCSP is likely to be an anomaly of our introduction of infinite signatures. Such a QCSP is unlikely to exist with a finite signature (at least, nothing like this is hitherto known). Indeed, its presence might be used as an argument against the acceptance of infinite signatures, if it is interpreted as an aberration. For the reader in this mind, we ask to please review the earlier paean to infinite signatures.

\section*{Acknowledgements}

We thank Hubie Chen and Micha\l\ Wrona for many useful discussions, as well as two anonymous referees for their advising on a previous draft.



\section*{Appendix A: Switchability and Collapsibility of Gap Algebras}

Let $f$ be a $k$-ary idempotent operation on domain $D$. We say $f$ is a \emph{generalised Hubie-pol} on $z_1\ldots z_k$ if, for each $i \in k$, $f(D,\ldots,D,z_i,D,\ldots,D)=D$ ($z_i$ in the $i$th position). When $z_1=\ldots=z_k=a$ this is called a \emph{Hubie-pol} in $\{a\}$ and gives $(k-1)$-Collapsibility from source $\{a\}$. In general, a generalised Hubie-pol does not bestow Collapsibility (\mbox{e.g.} Chen's $4$-ary Switchable operation $r$, below).  The name Hubie operation was used in \cite{LICS2015} for Hubie-pol and the fact that this leads to Collapsibility is noted in \cite{hubie-sicomp}.

For this appendix  $\mathbb{A}$ is an idempotent algebra on a $3$-element domain $\{0,1,2\}:=D$. Assume $\mathbb{A}$ has precisely two subalgebras on domains $\{0,2\}$ and $\{1,2\}$ and contains the idempotent semilattice-without-unit operation $s$ which maps all tuples off the diagonal to $2$. Thus, $\mathbb{A}$ is a \emph{Gap Algebra} as defined in \cite{AU-Chen-PGP}. Note that the presence of $s$ removes the possibility to have a $G$-set as a factor. We say that $\mathbb{A}$ is $\{0,2\}\{1,2\}$-projective if for each $k$-ary $f$ in $\mathbb{A}$ there exists $i \leq k$ so that, if $x_i \in \{0,2\}$ then $f(x_1,\ldots,x_k) \in \{0,2\}$ and if $x_i \in \{1,2\}$ then $f(x_1,\ldots,x_k) \in \{1,2\}$. Let us now further assume that $\mathbb{A}$ is not $\{0,2\}\{1,2\}$-projective. This rules out the Gap Algebras that have EGP and we now know that $\mathbb{A}$ is Switchable \cite{AU-Chen-PGP}. We will now consider the $4$-ary operation $r$ defined by Chen in \cite{AU-Chen-PGP}. Let $r$ be the idempotent operation satisfying
\[
\begin{array}{ccc}
0111 & & 1 \\
1011 & r & 1 \\
0001 &\mapsto & 0 \\
0010 & & 0 \\
\mbox{else} & & 2.
\end{array} 
\]
Chen proved that $(D;r,s)$ is $2$-Switchable but not $k$-Collapsible, for any $k$ \cite{AU-Chen-PGP}. Let $f$ be a $k$-ary operation in $\mathbb{A}$ that is not $\{0,2\}\{1,2\}$-projective. Violation of $\{0,2\}\{1,2\}$-projectivity in $f$ means that for each $i \in [k]$ either 
\begin{itemize}
\item there is $x_i \in \{0,1\}$ and $x_1,\ldots,x_{i-1},x_{i+1},\ldots,x_k \in \{0,1,2\}$ so that $f(x_1,\ldots,x_k)=y \in (\{0,1\}\setminus \{x_i\})$, or
\item or $x_i=c$ and there is $x_1,\ldots,x_{i-1},x_{i+1},\ldots,x_k \in \{0,1,2\}$ so that $f(x_1,\ldots,x_k)=y \in \{0,1\}$.
\end{itemize}
Note that we can rule out the latter possibility and further assume $x_1,\ldots,x_{i-1},$ $x_{i+1},\ldots,x_k \in \{0,1\}$, by replacing $f$ if necessary by the $2k$-ary $f(s(x_1,x'_1),\ldots,$ $s(x_k,x'_k))$. Thus, we may assume that (*) for each $i \in [k]$ there is $x_i \in \{0,1\}$ and $x_1,\ldots,x_{i-1},x_{i+1},\ldots,x_k \in \{0,1\}$ so that $f(x_1,\ldots,x_k)=y \in (\{0,1\}\setminus \{x_i\})$.

We wish to partition the $k$ co-ordinates of $f$ into those for which violation of $\{0,2\}\{1,2\}$-projectivity, on words in $\{0,1\}^k$:
\begin{itemize}
\item[$(i)$] happens with $0$ to $1$ but never $1$ to $0$.
\item[$(ii)$] happens with $1$ to $0$ but never $0$ to $1$.
\item[$(iii)$] happens on both $0$ to $1$ and $1$ to $0$.
\end{itemize}
Note that Classes $(i)$ and $(ii)$ are both non-empty (Class $(iii)$ can be empty). This is because if Class $(i)$ were empty then  $f(s(x_1,x'_1),\ldots,s(x_k,x'_k))$ would be a Hubie-pol in $\{1\}$ and if Class $(ii)$ were empty we would similarly have a Hubie-pol in $\{0\}$. We will write $k$-tuples with vertical bars to indicate the split between these classes. Suppose there exists a $\overline{z}$ so that $f(0,\ldots,0|1,\ldots,1|\overline{z}) \in \{0,1\}$. Then we can identify all the variables in one among Class $(i)$ or Class $(ii)$ to obtain a new function for which one of these classes is of size one. Note that if, e.g., Class $(i)$ is made singleton, this process may move variables previously in Class $(iii)$ into Class $(ii)$, but never to Class $(i)$. 

Thus we may assume that either Class $(i)$ or Class $(ii)$ is singleton or, for all $\overline{z}$ over $\{0,1\}$, $f(0,\ldots,0|1,\ldots,1|\overline{z}) =2$. Indeed, these singleton cases are dual and thus \mbox{w.l.o.g.} we need only prove one of them. Recall the global assumptions are in force for the remainder of the paper.

\subsection{Properties of Gap Algebras that are Switchable}

\begin{lemma}
\label{lem:fun}
Any algebra over $D$ containing $f$ and $s$ is either Collapsible or has binary term operations $p_1$ and $p_2$ so that 
\begin{itemize}
\item $p_1(0,1)=1$ and $p_1(1,0)=p_1(2,0)=2$, \textbf{and}
\item $p_2(0,1)=0$ and $p_2(1,0)=p_2(1,2)=2$.
\end{itemize}
\end{lemma}
\begin{proof}
Consider a tuple $\overline{x}$ over $\{0,1\}$ that witnesses the breaking of $\{0,2\}\{1,2\}$-projectivity for some Class $(i)$ variable from $0$ to $1$; so $f(\overline{x})=1$. Let $\widetilde{x}$ be $\overline{x}$ with the $0$s substituted by $2$ and the $1$s substituted by $0$. If, for each such  $\overline{x}$ over $\{0,1\}$ that witnesses the breaking of $\{0,2\}\{1,2\}$-projectivity for each Class $(i)$ variable, we find $f(\widetilde{x})=0$, then $f(s(x_1,x'_1),\ldots,s(x_k,x'_k))$ is a Hubie-pol in $\{1\}$. Thus, for some such  $\overline{x}$ we find $f(\widetilde{x})=2$. By collapsing the variables according to the division of $\overline{x}$ and $\widetilde{x}$ we obtain a binary function $p_1$ so that $p_1(0,1)=1$ and $p_1(2,0)=2$. We may also see that $p_1(1,0)=2$, since Classes $(i)$ and $(ii)$ are non-empty.

Dually, we consider tuples $\overline{x}$ over $\{0,1\}$ that witnesses the breaking of $\{0,2\}\{1,2\}$-projectivity for Class $(ii)$ variables from $1$ to $0$ to derive a function $p_2$ so that  $p_2(0,1)=0$, $p_2(1,2)=p(1,0)=2$.
\end{proof}

\subsubsection{The asymmetric case: Class $(i)$ is a singleton and there exists $\overline{z} \in \{0,1\}^*$ so that $f(0|1,\ldots,1|\overline{z})=1$}

We will address the case in which Class $(i)$ is a singleton and there exists $\overline{z} \in \{0,1\}^*$ so that $f(0|1,\ldots,1|\overline{z})=1$ (the like case with Class $(ii)$ being singleton itself being dual).

\begin{proposition}
\label{prop:asymmetric}
Let $f$ be so that Class $(i)$ is a singleton and there exists $\overline{z} \in \{0,1\}^*$ so that $f(0|1,\ldots,1|\overline{z})=1$. Then, either $f$ generates a binary idempotent operation with $01 \mapsto 0$ and $02 \mapsto 2$, or any algebra on $D$ containing $f$ and $s$ is Collapsible.
\end{proposition}
\begin{proof}
Let us consider the general form of $f$,
\[
\begin{array}{c|ccc|ccccc}
0& 1 & \cdots & 1 & z_0^0 & \cdots & z_0^{\ell'} & \ & 1 \\
0 & y^1_1 & \cdots & y^{k'}_1 & z_1^1 & \cdots & z_1^{\ell'} & \ & 0 \\
\vdots & \vdots & \cdots & \vdots & \vdots & \cdots & \vdots & \mapsto & \vdots \\
0 & y^1_{m'} & \cdots & y^{k'}_{m'} & z_{m'}^1 & \cdots & z_{m'}^{\ell'} & \ & 0 \\
\end{array}
\]
where the $y$s and $z$s are from $\{0,1\}$ and we can assume that each $(y^i_1,\ldots,y^i_{m'})$ contains at least one $1$ and also each $(z^i_1,\ldots,z^i_{m'})$ contains at least one $1$. For the latter assumption recall that in Class $(iii)$ we can always find some break of $\alpha\beta$-projectivity from $1$ to $0$. Note that by expanding what we previously called Class $(ii)$ we can build, by possibly identifying variables, a function $f'$ of the form
\[
\begin{array}{c|ccc|ccccc}
0& 1 & \cdots & 1 & 0 & \cdots & 0 & \ & 1 \\
0 & y^1_1 & \cdots & y^{k}_1 & z_1^1 & \cdots & z_1^{\ell} & \ & 0 \\
\vdots & \vdots & \cdots & \vdots  & \vdots & \cdots & \vdots & \mapsto & \vdots \\
0 & y^1_{m} & \cdots & y^{k}_{m} & z_{m}^1 & \cdots & z_{m}^{\ell} & \ & 0 \\
\end{array}
\]
where the $y$s and $z$s are from $\{0,1\}$ and we can assume each $(y^i_1,\ldots,y^i_{m})$ contains at least one $1$ and also each $(z^i_1,\ldots,z^i_{m})$ contains a least one $1$. Note that we do not claim the new vertical bars correspond to delineate between Classes $(i)$, $(ii)$ and $(iii)$ under their original definitions, since this is not important to us. We will henceforth assume that $f$ is in the form of $f'$.

Let $x^i_j$ (resp., $v^i_j$) be $0$ if $y^i_j$ (resp., $z^i_j$) is $0$, and be $2$ if $y^i_j$ (resp., $z^i_j$) is $1$. That is, $(x^1_{j}, \ldots,  x^{k}_{j}, v_{j}^1\ \ldots, v_{j}^{\ell})$ is built from $(y^1_{j}, \ldots,  y^{k}_{j}, z_{j}^1\ \ldots, z_{j}^{\ell})$ by substituting $1$s by $2$s. Suppose one of $f(0|x^1_1,\ldots,x^k_1| v_1^1\ \ldots, v_1^{\ell})$, \ldots, $f(0|x^1_m,\ldots,x^k_m | v_{m}^1\ \ldots, v_{m}^{\ell})$ is $2$. Then $f$ generates an idempotent binary operation with $01 \mapsto 0$ and $02 \mapsto 2$. Thus, we may assume that each of  $f(0|x^1_1,\ldots,x^k_1 | v_1^1\ \ldots, v_1^{\ell})$, \ldots, $f(0|x^1_m,\ldots,x^k_m|$ $v_m^1\ \ldots, v_m^{\ell}))$ is $0$. We now move to consider some cases.

(Case 1: $\ell=0$, \mbox{i.e.} there is nothing to the right of the second vertical bar.) From adversaries of the form $(\{0\}^M)$ and $(\{0,1\}^{m-1},\{1\}^{M-m+1})$ this supports construction of $(\{0,1\}^{m},\{1\}^{M-m})$ and all co-ordinate permutations. We illustrate this with the following diagram which makes some assumptions about the locations of the $1$s in each $(y^i_1,\ldots,y^i_{m})$; nonetheless it should be clear that the method works in general since there is at least one $1$ in $(y^i_1,\ldots,y^i_{m})$.
\[
\begin{array}{c|cccccc}
\{0\} & \{0,1\} & \{0,1\} & \cdots & \{0,1\} & & \{0,1\} \\
\vdots & \vdots & \vdots & \cdots & \vdots & & \{0,1\} \\
\{0\} & \{0,1\} & \{0,1\} & \cdots & \{0,1\} & & \{0,1\} \\
\{0\} & \{1\} & \{0,1\} & \cdots & \{0,1\} & & \{0,1\} \\
\{0\} & \{0,1\}  & \{1\} & \cdots & \{0,1\} & & \{0,1\} \\
\vdots & \vdots & \vdots & \cdots & \vdots & \mapsto & \{0,1\} \\
\{0\} & \{0,1\} & \{0,1\} & \cdots & \{1\} & & \{0,1\} \\
\{0\} & \{1\} & \{1\} & \cdots & \{1\} & & \{1\} \\
\vdots & \vdots & \vdots & \cdots & \vdots & & \{1\} \\
\{0\} & \{1\} & \{1\} & \cdots & \{1\} & & \{1\} \\
\end{array}
\]
Applying $s$ it is clear that the full adversary may be built from, for example, $(D^{m-1},\{0\}^{M-m+1})$ and $(D^{m-1},\{b\}^{M-m+1})$ which demonstrates $(m-1)$-Collapsibility.

(Case 2: $\ell\geq 1$.) Here we consider what is $f(0|1,\ldots,1|1,\ldots,1)$. If this is $1$ then we can clearly reduce to the previous case. If it is $0$ then $f(s(x_1,x_1'),\ldots,$ $s(x_{k+\ell+1},x'_{k+\ell+1}))$ is a generalised Hubie-pol in both $00|11,\ldots,11|00,\ldots,00$ and $00|11,\ldots,11|11,\ldots,11$, and we are Collapsible. This is because the composed function on these listed tuples gives $1$ and $0$, respectively, thus permitting to build adversaries of the form $(\{0,1\}^{k+\ell+2},\{0\}^{M-k-\ell-2})$ and  $(\{0,1\}^{k+\ell+2},\{1\}^{M-k-\ell-2})$ from adversaries of the form $(\{0,1\}^{k+\ell+1},\{0\}^{M-k-\ell-1})$ and  $(\{0,1\}^{k+\ell+1},$ $\{1\}^{M-k-\ell-1})$ (\mbox{cf.} Case 1).

Thus, we may assume $f(0|1,\ldots,1|1,\ldots,1)=2$. Using the fact that $f(s(x_1,x_1'),$ $\ldots,s(x_{k+\ell+1},x'_{k+\ell+1}))$ is a generalised Hubie-pol in $00|11\ldots 11|00 \ldots 00$ we can build (using $s$ and rather like in Case 1), from adversaries of the form $(\{0\}^{M})$ and  $(D^{(m-1)i},\{1\}^{M-(m-1)i})$, adversaries of the form $(D^{mi},\{1\}^{M-mi})$, and all co-ordinate permutations of this. Similarly, using the fact that $f(s(x_1,x_1'),\ldots,$ $s(x_{k+\ell+1},x'_{k+\ell+1}))$ is a generalised Hubie-pol in $00|11\ldots 11|11 \ldots 11$, we can build adversaries of the form $(D^{mi},\{2\}^{M-mi})$. 

(Case 2a: $f(0|2,\ldots,2|2,\ldots,2)=2$.) Consider again
\[
\begin{array}{c|ccc|ccccc}
0 & x^1_1 & \cdots & x^{k}_1 & v_1^1 & \cdots & v_1^{\ell} & \ & 0 \\
\vdots & \vdots & \cdots & \vdots  & \vdots & \cdots & \vdots & \mapsto & \vdots \\
0 & x^1_{m} & \cdots & x^{k}_{m} & v_{m}^1 & \cdots & v_{m}^{\ell} & \ & 0 \\
0 & 2 & \cdots & 2 & 2 & \cdots & 2 & & 2\\
\end{array}
\]
where each $(x^i_1,\ldots,x^i_{m})$ and $(v^i_1,\ldots,v^i_{m})$ contains at least one $2$. By amalgamating Classes $(ii)$ and $(iii)$ we obtain some function with the form
\[
\begin{array}{c|ccccc}
0 & u^1_1 & \cdots & x^{\nu}_1 & & 0 \\
\vdots & \cdots & \vdots & & & \vdots \\
0 & u^1_m & \cdots & x^{\nu}_m & & 0 \\
0 & 2 & \cdots & 2 & & 2\\
\end{array}
\]
where each $(u^i_1,\ldots,u^i_{m})$ is in $\{0,2\}^*$ and contains at least one $2$. From adversaries of the form $(D^{r+m-1}, \{2\}^{M-r-m+1})$ and $(\{0\}^{M})$ we can build $(D^{r},\{0,2\}^{M-r})$, and all co-ordinate permutations. We begin, pedagogically preferring to view some $D$s as $\{0,2\}$s,
\[
\begin{array}{c|cccccc}
\{0\} & D & D & \cdots & D & & D \\
\vdots & \vdots & \vdots & \cdots & \vdots & & D \\
\{0\} & D & D & \cdots & D & & D \\
\{0\} & \{2\}  &  \{0,2\} & \cdots &  \{0,2\} & & \{0,2\} \\
\{0\} &  \{0,2\}  & \{c\} & \cdots &  \{0,2\} & & \{0,2\} \\
\vdots & \vdots & \vdots & \cdots & \vdots & \mapsto & \{0,2\} \\
\{0\} &  \{0,2\}  &  \{0,2\} & \cdots & \{2\} & & \{0,2\} \\
\{0\} & \{2\}  & \{2\} & \cdots & \{2\} & & \{2\} \\
\vdots & \vdots & \vdots & \cdots & \vdots & & \{2\} \\
\{0\} & \{2\}  & \{2\} & \cdots & \{2\} & & \{2\} \\
\end{array}
\]
and follow with bottom parts of the form
\[
\begin{array}{c|ccccccccc}
\vdots & \vdots & \vdots & \cdots & \vdots & \mapsto & \{2\} \mbox{ or } \{0,2\} \\
\{0\} & \{0,2\}  & \{0,2\} & \cdots & \{0,2\} & & \{0,2\}. \\
\end{array}
\]
This now supports bootstrapping of the full adversary from adversaries of the form $(D^{m^2},\{0\}^{M-m^2})$, $(D^{m^2},$ $\{1\}^{M-m^2})$ and $(D^{m^2},\{2\}^{M-m^2})$. 

(Case 2b:  $f(0|2,\ldots,2|2,\ldots,2)=0$.) Here, from adversaries of the form $(D^{r+m-1}, \{2\}^{M-r-m+1})$ and $(\{0\}^{M})$ we can directly build $(D^{r+m-1}, \{0\}^{M-r-m+1})$.
\[
\begin{array}{c|cccccc}
\{0\} & D & D & \cdots & D  & & D \\
\vdots & \vdots & \vdots & \cdots & \vdots  & & D \\
\{0\} & D & D & \cdots & D  & & D \\
\{0\} & \{2\}  & \{2\} & \cdots & \{2\} & & \{0\} \\
\vdots & \vdots & \vdots & \cdots & \vdots  & \mapsto & \{0\} \\
\{0\} & \{2\}  & \{2\} & \cdots & \{2\}  & & \{0\} \\
\end{array}
\]
This now supports bootstrapping of the full adversary, similarly as in Case 2a (but slightly simpler).
\end{proof}

Let $\overline{x}:=x_1,\ldots,x_k$ and $\overline{y}:=y_1,\ldots,y_k$ be words over $\{0,1\} \ni x,y$. Let $\wedge(x,y)=0$ if $0 \in \{x,y\}$ and $1$ otherwise. Let $\vee(x,y)=1$ if $1 \in \{x,y\}$ and $0$ otherwise. This corresponds with considering $0$ as $\bot$ and $1$ as $\top$. Define $\wedge(\overline{x},\overline{y}):=(\wedge(x_1,y_1),\ldots,\wedge(x_k,y_k))$ and $\vee(\overline{x},\overline{y}):=(\vee(x_1,y_1),\ldots,\vee(x_k,y_k))$.  We are most interested in words 
\begin{itemize}
\item[A] $(\overline{x} | 0,\ldots,0 | \overline{z})$, such that $f(\overline{x} | 0,\ldots,0 | \overline{z})=0$, and for no $\overline{x}'\neq \overline{x}$ and $\overline{z}'$ over $\{0,1\}$ do we have  $(\overline{x}' | 0,\ldots,0 | \overline{z}')$ with $\vee(\overline{x},\overline{x}')=\overline{x}'$ so that $f(\overline{x}' | 0,\ldots,0 | \overline{z}')=0$.
\item[B] $(1,\ldots,1 | \overline{y} | \overline{z})$, such that $f(1,\ldots,1 | \overline{y} | \overline{z})=1$, and for no $\overline{y}'\neq \overline{y}$  and $\overline{z}'$ over $\{0,1\}$ do we have $(1,\ldots,1 | \overline{y}' | \overline{z}')$ with $\wedge(\overline{y},\overline{y}')=\overline{y}'$ so that $f(1,\ldots,1 | \overline{y}' | \overline{z}')=1$.
\end{itemize}
Such $\overline{x}$ and $\overline{y}$ are in a certain sense \emph{maximal}, but the sense of maximality is dual in Case B from Case A. $\overline{x}$ is maximal under inclusion for the number of $1$s it contains and $\overline{y}$ is maximal under inclusion for the number of $0$s it contains. In the asymmetric case that we consider here w.l.o.g., only Case A above will be salient, but we introduce both now for pedagogical reasons.

\begin{lemma}
\label{lem:r4-asymmetric}
Let $f$ be so that Class $(i)$ is a singleton and there exists $\overline{z} \in \{0,1\}^*$ so that $f(0|1,\ldots,1|\overline{z})=1$. Then any algebra over $D$ containing $f$ and $s$ is either Collapsible or has a $4$-ary term operation $r_4$ so that
\[
\begin{array}{ccc}
0101 & r_4 & 0 \\
0110 & \rightarrow & 0 \\
0111 & & 2 \\
\end{array}
\]
\end{lemma}
\begin{proof}
Recall $\exists \overline{z}$ so that $f(0|1,\ldots,1|\overline{z})=1$. Note that if exists $\overline{z}'$ over $\{0,1\}$ so that $f(0|1,\ldots,1|\overline{z}')=0$ then we have that $f(s(v_1,v'_1),\ldots,s(v_{k+\ell+1},v'_{k+\ell+1}))$ is a generalised Hubie-pol in both $11\ldots 11|00 \ldots 00|\widehat{z}$ and $11\ldots 11|00 \ldots 00|\widehat{z}'$, where we build widehat from overline by doubling each entry where it sits, and we become Collapsible. It therefore follows that there must exist distinct $\overline{y}_1$, $\overline{y}_2$, $\overline{z}_1$ and $\overline{z}_2$ (all over $\{0,1\}$) so that $f(0|\overline{y}_1|\overline{z}_1)=0$, $f(0|\overline{y}_2|\overline{z}_2)=0$ but $f(0|\vee(\overline{y}_1,\overline{y}_2)|\overline{z}_1)\neq 0$. By collapsing co-ordinates we get $f'$ so that
\[
\begin{array}{ccc}
0011 & f' & 0 \\
0101 &  \rightarrow & 0 \\
0111 & & \mbox{$0$ or $2$} \\
\end{array}
\]
The result follows by permuting co-ordinates, possibly in new combination through $s$ and the second co-ordinate.
\end{proof}

\subsubsection{The symmetric case: for every $\overline{z} \in \{0,1\}^*$ we have $f(0,\ldots,0|1,\ldots,1|\overline{z})=2$}

\begin{proposition}
\label{prop:symmetric}
Let $f$ be so that neither Class $(i)$ nor Class $(ii)$ is a singleton and so that for every $\overline{z} \in \{0,1\}^*$ we have $f(0,\ldots,0|1,\ldots,1|\overline{z})=2$. Then, either $f$ generates a binary idempotent operation with $01 \mapsto 0$ and $02 \mapsto 2$ or a binary idempotent operation with $01 \mapsto 1$ and $21 \mapsto 2$, or any algebra on $D$ containing $f$ and $s$ is Collapsible.
\end{proposition}
\begin{proof}
Let us consider the general form of $f$,
\[
\begin{array}{ccc|ccc|ccccc}
x^1_1 & \cdots & x^k_1 & 1 & \cdots & 1 & w_1^1 & \cdots & w_1^\ell & \ & 1 \\
\vdots & \vdots & \vdots & \vdots & \vdots & \vdots & \vdots & \vdots & \vdots & \mapsto & \vdots \\
x^1_m & \cdots & x^k_m & 1 & \cdots & 1 & w_m^1 & \cdots & w_m^\ell & \ & 1 \\
\\
0 & \cdots & 0 & y^1_1 & \cdots & y^\kappa_1 & z_1^1 & \cdots & z_1^\ell & \ & 0 \\
\vdots & \vdots & \vdots & \vdots & \vdots & \vdots & \vdots & \vdots & \vdots & \mapsto & \vdots \\
0 & \cdots & 0 & y^1_\mu & \cdots & y^\kappa_\mu & z_\mu^1 & \cdots & z_\mu^\ell & \ & 0 \\
\end{array}
\]
where the $x$s, $y$s, $z$s and $w$s are from $\{0,1\}$ and we can assume that each $(x^i_1,\ldots,x^i_m)$ and $(w^i_1,\ldots,w^i_m)$ contain at least one $0$ and $(y^i_1,\ldots,y^i_\mu)$ and $(z^i_1,\ldots,$ $z^i_\mu)$ contains at least one $1$. As in the previous proof we can make an assumption that each $(x^1_i, \cdots, x^k_i , 1, \ldots , 1 , w_i^1 , \ldots , w_i^\ell)$, with $0$ substituted for $2$, still maps under $f$ to $1$. Similarly,  each $(0,\ldots,0,y^1_i, \cdots, y^\kappa_i, z_i^1 , \ldots , z_i^\ell)$, with $1$ substituted for $2$, still maps under $f$ to $0$.

Since for each  $\overline{z} \in \{0,1\}^*$ we have $f(0,\ldots,0|1,\ldots,1|\overline{z})=2$ we can deduce that from the adversaries $(D^{(k+\kappa+\ell-1)i},\{0\}^{M-(k+\kappa+\ell-1)i})$ and  $(D^{(k+\kappa+\ell-1)i},$ $\{1\}^{M-(k+\kappa+\ell-1)i})$, adversaries of the form $(D^{(k+\kappa+\ell)i},\{2\}^{M-(k+\kappa+\ell)i})$, and all co-ordinate permutations of this. 

We now make some case distinctions based on whether $f(0,\ldots,0|2,\ldots,2|$ $2,\ldots,2)=2$ or $0$ and $f(2,\ldots,2|1,\ldots,1|2,\ldots,2)=2$ or $1$ (note that possibly Class $(iii)$ is empty). However, the method for building the full adversary from certain Collapsings proceeds very similarly to Cases 2a and 2b from Proposition~\ref{prop:symmetric}. We give an example below as to how, in the case $f(0,\ldots,0|2,\ldots,2|2,\ldots,2)=2$, we mimic Case 2a from Proposition~\ref{prop:asymmetric} to derive a function from this that builds, from adversaries of the form $(D^{r+m-1}, \{0\}^{M-(r+m-1)})$ and $(D^{r+2m-1},\{2\}^{M-(r+m-1)})$, we can build $(D^{r+m},\{0,2\}^{M-m-r})$. For pedagogic reasons we prefer to view some $D$s as $\{0,2\}$s,
\[
\begin{array}{cccc|cccccc}
\{0\} & D & \cdots & D & D & D & \cdots & D & & D \\
D & \{a\} & \cdots & D & D & D & \cdots & D & & D \\
\vdots & \vdots & \cdots & \vdots & \vdots & \vdots & \cdots & \vdots  & & D \\
D & D & \cdots & \{0\} & D & \cdots & D & D & & D \\
\{0\} & \{0\} & \cdots & \{0\}  & \{2\}  & \{0,2\} & \cdots & \{0,2\} & & \{0,2\} \\
\{0\}& \{0\}& \cdots &\{0\}& \{0,2\} & \{2\} & \cdots &\{0,2\} & &\{0,2\}\\
\vdots & \vdots & \vdots & \vdots & \vdots & \vdots & \cdots & \vdots  & \mapsto &\{0,2\} \\
\{0\}&\{0\} & \cdots &\{0\} &\{0,2\}& \{0,2\}& \cdots & \{2\} & & \{0,2\} \\
\{0\} & \{0\}& \cdots & \{0\}& \{2\}  &  \{2\} & \cdots & \{2\}  & & \{2\}  \\
\vdots & \vdots & \vdots & \vdots & \vdots & \vdots & \cdots & \vdots & &  \{2\}  \\
\{0\}&\{0\}& \cdots & \{0\}   &  \{2\} & \cdots & \{2\}  & &  \{2\} \\
\end{array}
\]
\end{proof}

\begin{lemma}
\label{lem:r4-symmetric}
Let $f$ be so that neither Class $(i)$ nor Class $(ii)$ is a singleton and so that for every $\overline{z} \in \{0,1\}^*$ we have $f(0|1,\ldots,1|\overline{z})=2$. Any algebra over $D$ containing $f$ is either Collapsible or contains a $4$-ary operations $r^a_4$ and $r^b_4$ with properties
\[
\begin{array}{ccc}
\begin{array}{ccc}
0101 & r^a_4 & 0 \\
0110 & \rightarrow & 0 \\
0111 & & 2 \\
\end{array}
& \ \ \ \ \ \  \mbox{\textbf{and}} \ \ \ \ \ \ &
\begin{array}{ccc}
0101 & r^b_4 & 1 \\
0110 & \rightarrow & 1 \\
0100 & & 2 \\
\end{array}
\end{array}
\]
\end{lemma}
\begin{proof}
The proof proceeds exactly as in Lemma~\ref{lem:r4-asymmetric}.
\end{proof}

An important special case  of the previous lemma, which is satisfied by Chen's $(\{0,1,2\};r,s)$ is as follows.

\vspace{0.2cm}
\noindent \textbf{Zhuk Condition}. $\mathbb{A}$ has idempotent term operations, binary $p$ and ternary operation $r_3$, so that either
\[
\begin{array}{c}
\left(
  \begin{array}{ccc}
    \begin{array}{ccc}
    001 & r_3 & 0 \\
    010 & \rightarrow & 0 \\
    011 & & 2 \\
    \end{array}
  & \ \ \ \ \ \  \mbox{\textbf{and}} \ \ \ \ \ \ &
    \begin{array}{ccc}
    01& p & 0 \\
    02 & \rightarrow & 2 \\
    \end{array}
  \end{array}
\right)
\\
\mbox{\textbf{or}} \\
\left(
\begin{array}{ccc}
   \begin{array}{ccc}
   101 & r_3 & 1 \\
  110  & \rightarrow & 1 \\
  100 & & 2 \\
   \end{array}
& \ \ \ \ \ \  \mbox{\textbf{and}} \ \ \ \ \ \ &
    \begin{array}{ccc}
    01 & p & 1 \\
    21 & \rightarrow & 2 \\
    \end{array}
  \end{array}
\right)
\end{array}
\]

\subsection{About essential relations}

We assume that all relations are defined on the finite set $\{0,1,2\}$.
A relation $\rho$ is called \textit{essential} if
it cannot be represented as a conjunction of relations with smaller arities.
A tuple $(a_{1},a_{2},\ldots,a_{n})$ is called \textit{essential for a relation $\rho$}
if $(a_{1},a_{2},\ldots,a_{n})\notin\rho$ and
for every $i\in \{1,2,\ldots,n\}$ there exists $b\in A$ such that 
$(a_{1},\ldots,a_{i-1},b,a_{i+1},\ldots,a_{n})\in\rho.$ Let us define a relation $\tilde{\rho}$ for every relation $\rho \subseteq D^n$. Put $\sigma_i(x_1,\ldots,x_{i-1},x_{i+1},\ldots,x_{n}) := \exists y \ \rho(x_1,\ldots,x_i,y,x_{i+1},\ldots,x_n)$
and let 
\[ \tilde{\rho}(x_1,\ldots,x_n) := \sigma_1(x_2,x_3,\ldots,x_n) \wedge  \sigma_2(x_1,x_3,\ldots,x_n) \wedge \ldots \wedge  \sigma_1(x_1,x_2,\ldots,x_{n-1}). \]
\begin{lemma}\label{sushnabor}
A relation $\rho$ is essential iff there exists an essential tuple for $\rho$.
\end{lemma}
\begin{proof}
(Forwards.) By contraposition, if $\rho$ is not essential, then $\tilde{\rho}$ is equivalent to $\rho$, and there can not be an essential tuple.

(Backwards.) An essential tuple witnesses that a relation is essential.
\end{proof}
\begin{lemma}
\label{lem:Dmitriy-micro}
Suppose $(2,2,x_3,\ldots,x_n)$ is an essential tuple for $\rho$. Then $\rho$ is not preserved by $s$.
\end{lemma}
\begin{proof}
Since $(2,2,x_3,\ldots,x_n)$ is an essential tuple, $(x_1,c,x_3,\ldots,x_n)$ and $(c,x_2,x_3,$ $\ldots,x_n)$ are in $\rho$ for some $x_1$ and $x_2$. But applying $s$ now gives the contradiction.
\end{proof}
For a tuple $\mathbf{y}$, we denote its $i$th co-ordinate by $\mathbf{y}(i)$. For $n\geq 3$, we define the arity $n+1$ idempotent operation $f^a_n$ as follows
\[
\begin{array}{c}
f^a_n(0,0\ldots,0,0)=0 \\
f^a_n(1,1,\ldots,1,1)=1 \\
f^a_n(1,0,\ldots,0,0)=0 \\
f^a_n(0,1,\ldots,0,0)=0 \\
\vdots \\
f^a_n(0,0,\ldots,1,0)=0 \\
f^a_n(0,0,\ldots,0,1)=0 \\
\mbox{else $2$}
\end{array}
\]
We define $f^b_n$ similarly with $0$ and $1$ swapped. These functions are very similar to partial near-unanimity functions.
\begin{lemma}
\label{lem:Dmitriy}
Suppose $\mathbb{A}$ is a Gap Algebra, that is not $\alpha\beta$-projective, so that $\mathbb{A}$ satisfies the Zhuk Condition. Then either 
\begin{itemize}
\item any relation $\rho \in \mathrm{Inv}(\mathbb{A})$ of arity $h<n+1$ is preserved by $f^a_n$, \textbf{or} 
\item any relation $\rho \in \mathrm{Inv}(\mathbb{A})$ of arity $h<n+1$ is preserved by $f^b_n$.
\end{itemize}
\end{lemma}
\begin{proof}
Suppose \mbox{w.l.o.g.} that the Zhuk Condition is in the first regime and has idempotent term operations, binary $p$ and ternary operation $r_3$, so that
\[
  \begin{array}{ccc}
    \begin{array}{ccc}
    001 & r_3 & 0 \\
    010 & \rightarrow & 0 \\
    011 & & 2 \\
    \end{array}
  & \ \ \ \ \ \  \mbox{\textbf{and}} \ \ \ \ \ \ &
    \begin{array}{ccc}
   01 & p & 0 \\
    02 & \rightarrow & 2 \\
    \end{array}
  \end{array}
\]
We prove this statement for a fixed $n$ by induction on $h$. For $h = 1$ we just need to
check that $f_n:=f^a_n$ preserves the unary relations $\{0,2\}$ and $\{1,2\}$.

Assume that $\rho$ is not preserved by $f_n$, then there exist tuples $\mathbf{y}_1,\ldots,\mathbf{y}_{n+1} \in \rho$ such that $f_n(\mathbf{y}_1,\ldots,\mathbf{y}_{n+1})=\gamma \notin \rho$. We consider a matrix whose columns are $\mathbf{y}_1,\ldots,\mathbf{y}_{n+1}$. Let the rows of this matrix be $\mathbf{x}_1,\ldots,\mathbf{x}_h$.

By the inductive assumption every $\sigma_i$ from the definition of $\widetilde{\rho}$ is preserved by $f_n$, which means that $\widetilde{\rho}$ is preserved by $f_n$, which means that $\gamma \notin \rho$ and $\gamma$ is an essential tuple for $\rho$.

We consider two cases. First, assume that $\gamma$ doesn't contain $2$. Then it follows from the definition that every $\mathbf{x}_i$ contains at most one element that differs from $\gamma(i)$. Since $n+1>h$, there exists $i \in \{1, 2, \ldots , n + 1\}$ such that $\mathbf{y}_i = \gamma$. This contradicts the fact that $\gamma \notin \rho$.

Second, assume that $\gamma$ contains $2$. Then by Lemma~\ref{lem:Dmitriy-micro}, $\gamma$ contains exactly one $2$. \mbox{W.l.o.g.} we assume that $\gamma(1) = 2$. It follows from the definition of $f_n$ that $\mathbf{x}_i$ contains at most one element that differs from $\gamma(i)$ for every $i \in \{2, 3, \ldots , h\}$. Hence, since $n+1>h$, for some $k \in \{1, 2, \ldots , n+ 1\}$ we have $\mathbf{y}_k(i) = \gamma(i)$ for every $i \in \{2, 3, \ldots , h\}$. Since $f_n(\mathbf{x}_1) = 2$, we have one of three subcases. First subcase, $\mathbf{x}_1(j) = 2$ for some $j$. We need one of the properties
\[
\begin{array}{cc|c}
\mathbf{y}_k & \mathbf{y}_j & \gamma \\
\hline
0& 2 & 2 \\
0 & 1 & 0 \\
\end{array}
\mbox{ \ \ \ \ \ \ \ \ \ \ \ \ \ \ \ \ \ \ }
\begin{array}{cc|c}
\mathbf{y}_k & \mathbf{y}_j & \gamma \\
\hline
1 & 2 & 2 \\
0 & 1 & 0 \\
\end{array}
\]
and we can see that the functions from Lemma~\ref{lem:fun} or the definition of the Zhuk Condition suffice, which contradicts our assumptions.

Second subcase, $\mathbf{y}_k(1) = 1, \mathbf{y}_m(1) = 0$ for some $m \in \{1, 2, \ldots , n + 1\}$. We need the property 
\[
\begin{array}{cc|c}
\mathbf{y}_k & \mathbf{y}_m & \gamma \\
\hline
1 & 0 & 2 \\
0 & 1 & 0 \\
\end{array}
\]
can check that a function from Lemma~\ref{lem:fun} suffices, which contradicts our assumptions.

Third subcase, $\mathbf{y}_k(1) = 0, \mathbf{y}_m(1) = 1$ and $\mathbf{y}_l(1) = 1$ for $m, l \in \{1, 2, \ldots , n + 1\}\setminus \{k\}$, $m \neq l$. We need the property
\[
\begin{array}{ccc|c}
\mathbf{y}_k & \mathbf{y}_m & \mathbf{y}_l & \gamma \\
\hline
0 & 1 & 1 & 2 \\
0 & 0 & 1 & 0 \\
0 & 1 & 0 & 0 \\
\end{array}
\]
and we can check that the $r_3$ from the Zhuk Condition suffices, which contradicts our assumptions. This completes the proof.
\end{proof}
\begin{corollary}
\label{cor:Dmitriy}
Suppose $\mathbb{A}$ is a Gap Algebra, that is not $\alpha\beta$-projective so that $\mathbb{A}$ satisfies the Zhuk Condition. Then, for every finite subset of $\Delta$ of Inv$(\mathbb{A})$, Pol$(\Delta)$ is Collapsible.
\end{corollary}
\begin{proof}
$f^a_n$ is a Hubie-pol in $\{1\}$ and $f^b_n$ is a Hubie-pol in $\{0\}$.
\end{proof}
For $n\geq 2$, we define the arity $n+2$ idempotent operation $\widehat{f}^a_n$ as follows
\[
\begin{array}{c}
f^a_n(0,0,0,\ldots,0,0)=0 \\
f^a_n(1,1,1,\ldots,1,1)=1 \\
f^a_n(1,1,0,\ldots,0,0)=0 \\
f^a_n(1,0,1,\ldots,0,0)=0 \\
\vdots \\
f^a_n(1,0,0,\ldots,1,0)=0 \\
f^a_n(1,0,0,\ldots,0,1)=0 \\
\mbox{else $c$}
\end{array}
\]
We define $\widehat{f}^b_n$ similarly with $0$ and $1$ swapped.
\begin{lemma}
\label{lem:Dmitriy-long}
Suppose $\mathbb{A}$ is a Gap Algebra that is not $\alpha\beta$-projective. Then either 
\begin{itemize}
\item any relation $\rho \in \mathrm{Inv}(\mathbb{A})$ of arity $h<n+2$ is preserved by $\widehat{f}^a_n$, \textbf{or} 
\item any relation $\rho \in \mathrm{Inv}(\mathbb{A})$ of arity $h<n+2$ is preserved by $\widehat{f}^b_n$.
\end{itemize}
\end{lemma}
\begin{proof}
Suppose \mbox{w.l.o.g.} that we are either in the asymmetric case with Class $(i)$ singleton and  there exists $\overline{z} \in \{0,1\}^*$ so that $f(0|1,\ldots,1|\overline{z})=1$ \textbf{or} we are in the symmetric case and we have an idempotent term operation $p$ mapping $01 \mapsto 0$ and $02 \mapsto 2$.

We prove this statement for a fixed $n$ by induction on $h$. For $h = 1$ we just need to
check that $\widehat{f}_n:=\widehat{f}^a_n$ preserves the unary relations $\{0, 2\}$ and $\{1, 2\}$.

Assume that $\rho$ is not preserved by $f_n$, then there exist tuples $\mathbf{y}_1,\ldots,\mathbf{y}_{n+2} \in \rho$ such that $\widehat{f}_n(\mathbf{y}_1,\ldots,\mathbf{y}_{n+2})=\gamma \notin \rho$. We consider a matrix whose columns are $\mathbf{y}_1,\ldots,\mathbf{y}_{n+2}$. Let the rows of this matrix be $\mathbf{x}_1,\ldots,\mathbf{x}_h$.

By the inductive assumption every $\sigma_i$ from the definition of $\widetilde{\rho}$ is preserved by $\widehat{f}_n$, which means that $\widetilde{\rho}$ is preserved by $\widehat{f}_n$, which means that $\gamma \notin \rho$ and $\gamma$ is an essential tuple for $\rho$.

We consider two cases. First, assume that $\gamma$ doesn’t contain $2$. Then it follows from the definition that every $\mathbf{x}_i$ contains at most one element that differs from $\gamma(i)$. Since $n+2>h$, there exists $i \in \{1, 2, \ldots , n + 1\}$ such that $\mathbf{y}_i = \gamma$. This contradicts the fact that $\gamma \notin \rho$.

Second, assume that $\gamma$ contains $2$. Then by Lemma~\ref{lem:Dmitriy-micro}, $\gamma$ contains exactly one $2$. \mbox{W.l.o.g.} we assume that $\gamma(1) = 2$. It follows from the definition of $\widehat{f}_n$ that $\mathbf{x}_i$ contains at most one element that differs from $\gamma(i)$ for every $i \in \{2, 3, \ldots , h\}$. Hence, since $n+2>h$, for some $k \in \{2, \ldots , n+2\}$ we have $\mathbf{y}_k(i) = \gamma(i)$ for every $i \in \{2, 3, \ldots , h\}$. Since $\widehat{f}_n(\mathbf{x}_1) = 2$, we have one of four subcases. First subcase, $\mathbf{x}_1(j) = 2$ for some $j$. We need one of the properties
\[
\begin{array}{cc|c}
\mathbf{y}_k & \mathbf{y}_j & \gamma \\
\hline
0 & 2 & 2 \\
0 & 1 & 0 \\
\end{array}
\mbox{ \ \ \ \ \ \ \ \ \ \ \ \ \ \ \ \ \ \ }
\begin{array}{cc|c}
\mathbf{y}_k & \mathbf{y}_j & \gamma \\
\hline
1 & 2 & 2 \\
0 & 1 & 0 \\
\end{array}
\]
and we can see that the functions from Lemma~\ref{lem:fun}, or Proposition~\ref{prop:asymmetric} or Proposition~\ref{prop:symmetric}, suffice which contradicts our assumptions.

Second subcase, $\mathbf{y}_k(1) = 1, \mathbf{y}_m(1) = 0$ for some $m \in \{1, 2, \ldots , n + 1\}$. We need the property 
\[
\begin{array}{cc|c}
\mathbf{y}_k & \mathbf{y}_m & \gamma \\
\hline
1 & 0 & 2 \\
0 & 1 & 0 \\
\end{array}
\]
can check that a function from Lemma~\ref{lem:fun} suffices, which contradicts our assumptions.

For Case 3, $\mathbf{y}_k(1) = 0, \mathbf{y}_m(1) = 1$ and $\mathbf{y}_l(1) = 1$ for some $m, l \in \{1, 2, \ldots , n + 1\}\setminus \{k\}$, $m \neq l$ (possibly $1 \in \{m,l\}$). We now split into two subsubcases: either $\mathbf{y}_1(1)=1$ and we need the property
\[
\begin{array}{cccc|c}
\mathbf{y}_1 &  \mathbf{y}_k & \mathbf{y}_m & \mathbf{y}_l & \gamma \\
\hline
1 & 0 & 1 & 1 & 2 \\
1 & 0 & 0 & 1 & 0 \\
1 & 0 & 1 & 0 & 0. \\
\end{array}
\]
Here we can check that $r_4$, from Proposition~\ref{prop:asymmetric} or Proposition~\ref{prop:symmetric}, with co-ordinates $1$ and $2$ permuted, suffices, which contradicts our assumptions. Or we have  $\mathbf{y}_1(1)=0$ and we need the property
\[
\begin{array}{cccc|c}
\mathbf{y}_1 & \mathbf{y}_k & \mathbf{y}_m & \mathbf{y}_l & \gamma \\
\hline
0 & 0 & 1 & 1 & 2 \\
1 & 0 & 0 & 1 & 0 \\
1 & 0 & 1 & 0 & 0. \\
\end{array}
\]
For this $p(x_2,(p_1(x_4,p_1(x_2,x_1))))$ suffices where $p_1$ comes from Lemma~\ref{lem:fun} and $p$ is as before in this proof (\mbox{cf.} Proposition~\ref{prop:asymmetric} and Proposition~\ref{prop:symmetric}).
\[
\begin{array}{cccc|c|c|c}
x_1 & x_2 & x_3 & x_4 & p_1(x_2,x_1) & p_1(x_4,p_1(x_2,x_1)) & p(x_2,(p_1(x_4,p_1(x_2,x_1)))) \\
\hline
0 & 0 & 1 & 1 & 0 & 2 & 2 \\
1 & 0 & 0 & 1 & 1 & 1 & 0\\
1 & 0 & 1 & 0 & 1 & 1 & 0 \\
\end{array}
\]
This completes the proof.
\end{proof}

Let us recall the main result of this appendix.

\

\noindent \textbf{Theorem~\ref{cor:Dmitriy-long}}.
Suppose $\mathbb{A}$ is a Gap Algebra that is not $\alpha\beta$-projective. Then, for every finite subset of $\Delta$ of Inv$(\mathbb{A})$, Pol$(\Delta)$ is Collapsible.

\begin{proof}
$\widehat{f}^a_n$ is a Hubie-pol in $\{1\}$ and $\widehat{f}^b_n$ is a Hubie-pol in $\{0\}$.
\end{proof}

\section*{Appendix B: $\mathbb{A}$ has a $2$-element G-set as a subalgebra}

Let us recall that this appendix is in pursuit of the following result.

\

\noindent \textbf{Theorem~\ref{thm:catarina}}.
Suppose $\mathbb{A}$ is a $3$-element algebra  that is not $\alpha\beta$-projective, containing a $2$-element G-set as a subalgebra. Then, $\mathbb{A}$ is Collapsible.

\

\noindent Recall $\mathbb{A}$ is an idempotent clone over domain $\{0,1,2\}:=D$, without loss of generality, having a subalgebra induced by $\{0,1\}$ that is a G-set. Further we can assume that $\mathbb{A}$ is neither: $\{0,2\}\{1,2\}$-projective, $\{0,1\}\{1,2\}$-projective nor $\{0,1\}\{0,2\}$-projective. From this last assumption, by collapsing co-ordinates of $4$-ary operations, we arrive at the following.

\begin{lemma}\label{4hypothesis}
$\mathbb{A}$ has three ternary operations $f,g,h$ all of which are projections to the first component on $\{0,1\}$ and for which:
\[
\begin{array}{lll}
\mbox{one of: \ \ \ \ \ \ \ \ \ \ \ \ \ } & \mbox{one of: \ \ \ \ \ \ \ \ \ \ \ \ \ }  & \mbox{one of: \ \ \ \ \ \ \ \ \ \ \ \ \ }  \\
f(0,1,2)=1 & g(1,0,2)=0 & h(0,1,2)=1 \\
f(1,0, 2)=0 & g(1, 0,2)=2 & h(0,1,2)=2 \\
f(2,0,1)=1 & g(0,1,2)=2 & h(1,0,2)=2\\
f(2,0,1)=0 & g(2,0,1)=0 & h(2,0,1)=1
\end{array}
\]
\end{lemma}


\noindent Our proof proceeds by a lengthy case analysis.
We first look at the cases when $\mathbb{A}$ does not have a $G$-set as a homomorphic image (otherwise $\mathbb{A}$ has EGP). There are 3 possible congruences that originate from the homomorphic image, they are   $\{0, 2\}, \{1\}$, $\{1,2\}, \{0\}$ and $\{0,1\}, \{2\}$.  We will consider the case where no such congruences exist and the case when even if they exist they do not yield a $G$-set.
 
 We start by noting that since $\mathbb{A}$  has ternary operations as mentioned in Lemma \ref{4hypothesis}  some congruences can be discarded:
 
 \begin{claim}
 $\{1,2\}, \{0\}$  is only a congruence of $\mathbb{A}$ if $f(2,0,1)=1$ and $g(1,0,2)=2$, and  $\{0, 2\}, \{1\}$ is only a congruence if $f(2,0,1)=0$ and $h(0,1,2)=2$.
 \end{claim}
 \begin{proof}
 Let $\rho$ denote the congruence $\{1,2\}, \{0\}$.  Let $\psi$ be any operation of $A$ that acts as the first projection on $\{0,1\}$. We have $$\psi(0,1,2) \rho\  \psi(0,1,1)\ \Leftrightarrow \ \psi(0,1,2)\rho \ 0$$ 
 $$\psi(2,0,1) \rho \ \psi(1,0,1)\ \Leftrightarrow \ \psi(2,0,1)\rho \ 1$$
  $$\psi(1,0,2) \rho \ \psi(1,0,1)\ \Leftrightarrow \ \psi(1,0,2)\rho \ 1$$ which implies that  $\psi(0,1,2)=0$ and $\psi(2,0,1), \psi(1,0,2)\in \{1,2\}$.  It follows, from the assumptions in Lemma \ref{4hypothesis} that we must have $f(2,0,1)=1$, $g(1,0,2)=2$, and $h(2,0,1)=1$ or $h(1,0,2)=2$.
  
  Now let $\theta$ denote the congruence $\{0, 2\}, \{1\}$. We have 
  $$\psi(0,1,2) \theta \  \psi(0,1,0)\ \Leftrightarrow \ \psi(0,1,2)\theta \ 0$$ 
 $$\psi(2,0,1) \theta \ \psi(0,0,1)\ \Leftrightarrow \ \psi(2,0,1)\theta \ 0$$
  $$\psi(1,0,2) \theta \ \psi(1,0,0)\ \Leftrightarrow \ \psi(1,0,2)\theta  \ 1$$ which implies that  $\psi(0,1,2), \psi (2,0, 1)\in \{0,2\}$ and $\psi(1,0,2)=1$. It follows, from the assumptions in Lemma \ref{4hypothesis} that, for $\theta$ to be a congruence,  we must have $f(2,0,1)=0$,  $h(0,1, 2)=2$, and  $g(0,1,2)=2$ or $g(2,0,1)=0$.
\end{proof}

It is then clear that  $\{1,2\}, \{0\}$ and  $\{0, 2\}, \{1\}$  cannot both be congruences of $\mathbb{A}$ simultaneously. From the operations $f, g$ and $h$ of $\mathbb{A}$ we cannot discard the congruence $\{0,1\}, \{2\}$ without considering different operations. We now look at the two possible cases: $\mathbb{A}$ has no homomorphic images (and so no congruences) and $\mathbb{A}$ has a homomorphic image but it is not a $G$-set.

\subsection{ $\mathbb{A}$ has no congruences}

Since we do not want $\{0,1\}, \{2\}$ to be a congruence of $\mathbb{A}$,  there must exist an operation $r$ on $\mathbb{A}$ 
such that 
$r(0,0,1,1,2)=2$, 
$r(0,1,0,1,2)\in \{0,1\}$ and $r$ is a projection on $\{0,1\}$. We consider 5 cases, depending on the possible projection, we rearrange $r$ such that it always behave as the first projection on $\{0,1\}$.


\subsubsection{Case 1} $r(0,0,1,1,2)=2$, 
$r(0,1,0,1,2)=1$, and $r$ is the first projection on $\{0,1\}$

We aim to prove that $r$ is $4$-Collapsible. We sometimes identify $j$ with $\{j\}$ when it is will not cause great confusion.

Note that our derivation is unlikely to be optimal since $r(0,D,D,D,D)=D$ whereas we use just $r(0,D,D,D,D)\supseteq \{0,2\}$ below.
\begin{lemma}
Let $k \geq 4$. From Adversaries that are co-ordinate permutations of $(D^k,\{0\}^{M-k})$, $(D^k,\{1\}^{M-k})$ and $(D^k,\{2\}^{M-k})$ one can build Adversaries that are co-ordinate permutations of the form $(D^k,\{0,2\}^{M-k})$.
\end{lemma}
\begin{proof}  
First perform
\[
\begin{array}{ccccccc}
0 & D & D & D & D & & \{0,2\} \\
D & 0 & D & D & D & &  D \\
D & D & 1 & D & D & & D \\
D & D & D & 1 & D & \rightarrow & D \\
D & D & D & D & 2 & & \{1,2\} \\
0 & 0 & 1 & 1 & 2 & & 2 \\
0 & 0 & 1 & 1 & 2 & & 2 \\
\vdots & \vdots & \vdots & \vdots & \vdots & & \vdots 
\end{array}
\]
then
\[
\begin{array}{ccccccc}
0 & D & D & D & D & & \{0,2\} \\
D & 0 & D & D & D & &  D \\
D & D & 1 & D & D & & D \\
D & D & D & 1 & 2 & \rightarrow & \{1,2\} \\
D & D & D & D & \{1,2\} & & D \\
0 & 0 & 1 & 1 & \{0,2\} & & \{0,2\} \\
0 & 0 & 1 & 1 & 2 & & 2 \\
\vdots & \vdots & \vdots & \vdots & \vdots & & \vdots 
\end{array}
\]
and
\[
\begin{array}{ccccccc}
0 & D & D & D & D & & \{0,2\} \\
D & 0 & D & D & D & &  D \\
D & D & 1 & D & D & & D \\
D & D & D & 1 & \{1,2\} & \rightarrow & D \\
D & D & D & D & \{1,2\} & & D \\
0 & 0 & 1 & 1 & \{0,2\} & & \{0,2\} \\
0 & 0 & 1 & 1 & 2 & & 2 \\
\vdots & \vdots & \vdots & \vdots & \vdots & & \vdots 
\end{array}
\]
Now repeat this process $M-k-2$ times.
\end{proof}
\begin{lemma}
Let $k \geq 4$. From Adversaries that are co-ordinate permutations of the form $(D^k,\{0,2\}^{M-k})$ and $(D^k,\{1\}^{M-k})$ one can build Adversaries that are co-ordinate permutations of the form $(D^k,\{0,1\}^{M-k})$.
\end{lemma}
\begin{proof}  
\[
\begin{array}{ccccccc}
\{0,2\} & D & D & D & D & & D \\
D & 1 & D & D & D  & &  \{0,1\} \\
D & D & \{0,2\} & D & D & \rightarrow & D \\
D & D & D & 1 & D & & D \\
D & D & D & D & \{0,2\} & & D \\
\{0,2\} & 1 & \{0,2\} & 1 & \{0,2\} & & \{0,1\} \\
\{0,2\} & 1 & \{0,2\} & 1 & \{0,2\} & & \{0,1\} \\
\vdots & \vdots & \vdots & \vdots & \vdots & & \vdots 
\end{array}
\]
\end{proof}
\begin{lemma}
Let $k \geq 4$. From Adversaries that are co-ordinate permutations of the form $(D^k,\{0,1\}^{M-k})$ and $(D^k,\{0,2\}^{M-k})$ one can build the full Adversarie $(D^m)$.
\end{lemma}
\begin{proof}  
\[
\begin{array}{ccccccc}
\{0,1\} & D & D & D & D & &  D \\
D & \{0,1\} & D & D & D & & D \\
D & D & \{0,1\} & D & D & \rightarrow & D \\
D & D & D & \{0,1\} & D & & D \\
D & D & D & D & \{0,2\} & & D \\
\{0,1\} & \{0,1\} & \{0,1\} & \{0,1\} & \{0,2\} & & D \\
\{0,1\} & \{0,1\} & \{0,1\} & \{0,1\} & \{0,2\} & & D \\
\vdots & \vdots & \vdots & \vdots & \vdots & & \vdots 
\end{array}
\]
\end{proof}\qed

\subsubsection{Case 2}  $r(0,0,1,1,2)=2$, 
$r(0,1,0,1,2)=0$, and $r$ is the first projection on $\{0,1\}$

\begin{enumerate}

\item 
We assume  first that $f(0,1,2)=1$.





Then we have the following operations on $\mathbb{A}$: 
\begin{claim}
The  
$45$-ary operation $c_1$ defined as 
$$\begin{array}{rll}

f( 
&  r(f(x_1, \ldots, x_3), f(x_4, \ldots, x_6),f(x_7, \ldots,x_{9}),f(x_{10}, \ldots, x_{12}),f(x_{13}, \ldots, x_{15})),\\
& \cdots\\
& r(f(x_{31}, \ldots, x_{33}),f(x_{34}, \ldots, x_{36}),f(x_{37}, \ldots, x_{39}),f(x_{40}, \ldots, x_{42}),f(x_{43}, \ldots, x_{45})) \ )

\end{array}$$
is a generalized Hubie-pol  and on the tuple $ 000012000111222 (\times 3)$
 it returns $0$.
 \end{claim}
 
 \begin{proof}
 When applying $c_1$ to the tuple above, recalling that $f$, $g$ and $r$ are idempotent, we obtain 
 $$ r(f(0,0,0), f(0, 1,2), f(0,0,0), g(1,1,1),f(2,2,2))=r(0,1,0,1,2)=0.$$
 Now let us check that the operation is a generalized Hubie-pol: we have $f(0, D, D),$ $f(D, 0, D), f(D,D, 0)\supseteq \{0,1\}$, and $r(\{0,1\}, D, \ldots, D)= D$, it follows that $c_1(0, D, \ldots, D)=c_1(D,0, D, \ldots, D)= c_1(D, D, 0, D \ldots, D)=D$ and we obtain the same result  when $c_1$ is applied to all $D$s and one 0, and the $0$ appears in co-ordinates congruent with $1, 2$ or $3$ modulo $15$; we have $f(0, D, D) ,f(D, 1, D)\supseteq \{0,1\}$, $f(D, D, 2)\supseteq \{1,2\}$, and $r(D, \{1,2\}, D, D, D)=D= r(D, \{0,1\}, D, D, D)$ it follows that when $c_1$ is applied to all $D$ and one $0$, $1$ or $2$ in co-ordinates congruent with $4, 5$ or $6$ modulo $15$, respectively,  the result is $D$; now $f(0, D, D), f(D, 0, D), f(D, D, 0)\supseteq \{0,1\}$, and $r(D, D, \{0,1\}, D, D)=D$, thus when $c_1$ is applied to all $D$ and one $0$ in co-ordinates congruent with $7, 8$ or $9$ modulo $15$ the result is $D$; we have $f(1, D, D)\supseteq \{1\}$, $f(D, 1, D), f(D, D, 1)\supseteq \{0,1\}$ and $r(D, D, D, 1, D)=D$, hence when $c_1$ is applied to all $D$ and one $1$ in co-ordinates congruent with $10, 11$ or $12$ modulo $15$ the result is $D$;  finally $f(2, D, D), f(D, 2, D), f(D, D, 2)\supseteq \{2\}$,  $r(D, D, D, D, 2)\supseteq \{0,2\}$, and $f(\{0,2\}, D, D)=f(D, \{0,2\}, D)=f(D, D, \{0,2\})=D$, thus when $c_1$ is applied to all $D$ and one $2$ in co-ordinates congruent with $13, 14 $ or $0$ modulo $15$ the result is $D$. This proves the claim.
  \end{proof}\qed

 
 \begin{claim}
The $45$-ary operation $c_2$ defined by 
 $$\begin{array}{rl}
f(& 
  r(f(x_1, \ldots, x_3), f(x_4, \ldots, x_6),f(x_7, \ldots,x_{9}),f(x_{10}, \ldots, x_{12}),f(x_{13}, \ldots, x_{15})),\\
& \cdots\\
 &r(f(x_{31}, \ldots, x_{33}),f(x_{34}, \ldots, x_{36}),f(x_{37}, \ldots, x_{39}),g(x_{40}, \ldots, x_{42}),f(x_{43}, \ldots, x_{45})) \ )
  \end{array}$$
  is a generalized Hubie-pol  on the elements 
$000000111111222(\times 3)$ and on this tuple it returns $2$.
\end{claim}
\begin{proof}
 When applying $c_2$ to the tuple above, recalling that $f$, $g$ and $r$ are idempotent, we obtain 
 $$ r(f(0,0,0), f(0, 0,0), f(1,1,1), f(1,1,1), f(2,2,2))=r(0,0,1,1,2)=2.$$
 Now let us check that the operation is a generalized Hubie-pol: we have $f(0, D, D),$ $f(D, 0, D), f(D, D, 0)\supseteq \{0,1\}$,  and $r(\{0,1\}, D, \ldots, D)= r(D, \{0,1\}, D, D, D)=D$, it follows that $c_2(0, D, \ldots, D)=D$ and we obtain the same result  when $c_2$ is applied to all $D$s and one 0, and the $0$ appears in co-ordinates congruent with $1, 2, 3, 4, 5$ or $6$ modulo $15$; 
we have $f(1, D, D) \supseteq \{1\}$, $f(D, 1, D), f(D, D, 1)\supseteq \{0,1\}$, and $r(D, D, \{1\}, D, D)= r(D, D, \{0,1\}, D, D)= r(D, D, D,\{1,2\},  D)= r(D, D, D,\{0,1\},  D) =D$ it follows that when $c_2$ is applied to all $D$ and one $1$ in co-ordinates congruent with $7, 8, 9, 10, 11$ or $12$ modulo $15$ the result is $D$;   finally $f(2, D, D), f(D, 2, D), f(D, D, 2)\supseteq \{2\}$,  $r(D, D, D, D, 2)\supseteq \{0,2\}$, and $f(\{0,2\}, D, D)=f(D, \{0,2\}, D)=f(D, D, \{0,2\})=D$, thus when $c_2$ is applied to all $D$ and one $2$ in co-ordinates congruent with $13, 14 $ or $0$ modulo $15$ the result is $D$. This proves the claim.

\end{proof}\qed

\begin{claim}
The   $55$-ary operation $c_3$ defined by 
$$\begin{array}{rl}r(& 
 f(f(x_1, \ldots, x_3),f(x_4, \ldots, x_6), r(x_{7}, \ldots, x_{11}),\\
 &  \cdots,\\
 &   f(f(x_{47}, \ldots, x_{49}),f(x_{50}, \ldots, x_{52}),r(x_{53}, \ldots, x_{55})) \ )\end{array}$$
 is a generalized Hubie-pol on the elements $ 000111 00112 (\times 5) $, returning $1$ on this tuple.
\end{claim}

\begin{proof}
We start by noting that if $r(2,2,2,2,0)=1$ then the operation 
$s(x_1, x_2, x_3, x_4, x_5)=
r(x5,x5,x5,x5,r(x1,x2,x3,x4,x5))$ satisfies 
$s(0,1,0,1,2) = 1$ and $s(0,0,1,1,2)= 2$, so we are in Case 1; if $r(2,2,2,2,0)=0$ then the operation 
$$f(r(x_1, \ldots, x_5), r(x_6, \ldots, x_{10}), r(x_{11},\ldots, x_{15}))$$ 
is a Hubie-pol on $2$.  Indeed we have $r(2, D, \ldots, D), \ldots, r(D, \ldots, D, 2)\supseteq \{0,2\}$,  and $f(\{0,2\}, D, D)=f(D, \{0,2\}, D)=f(D, D, \{0,2\})=D$.

So we assume now that $r(2, 2, 2, 2 ,0)=2$.

When applying $c_3$ to the tuple $00011100112(\times 5)$, recalling that $f$, $g$ and $r$ are idempotent, we obtain 
 $$f(f (0,0,0), g(1,1,1), r(0,0,1,1,2))=f(0,1,2)=1.$$
Now let us check that the operation is a generalized Hubie-pol: we have $f(0, D, D)$, $f(D, 0, D), f(D, D, 0) \supseteq \{0,1\}$,  and $r(\{0,1\}, D, \ldots, D)=\cdots= r(D, \ldots, D, \{0,1\})=D$, so when applying $c_3$ is applied to all $D$s and one $0$, and the $0$ appears in co-ordinates congruent with $1, 2 $ or $3$ modulo $11$ the result is $D$; 
then $f(1, D, D) \supseteq \{1\},\  f(D, 1, D), f(D, D, 1)\supseteq \{0,1\}$,  it follows then, as above, that when applying $c_3$ to all $D$s and one $1$, and the $1$ appears in co-ordinates congruent with $4, 5$ or $6$ modulo $11$ the result is $D$ 
;  we also have $r(0, D, \ldots, D), r(D, \ldots, D, 2)\supseteq \{0,2\}$, $r(D, 0, D, D, D)=r(D, D, 1, D, D)=r(D, D, D, 1, D)=D$, and $f(D, D, \{0,2\})=D$, hence when $c_3$ is applied to $D$ in 
all co-ordinates except one and that one co-ordinate is either a $0$ if the co-ordinate is congruent with  $7$ or $8$, a $1$ if the co-ordinate is congruent with $9$ or $10$, or a $2$ if the 
co-ordinate is congruent with $0$ modulo $11$ then the result is $D$. This proves the claim.
\end{proof}\qed

\item If $f(2,0,1)=1$ then defining an operation 
$ f'(x,y,z)=f(r(x, x, y, y, z), x, y)$ we have that $f'$ is the first projection on $\{0,1\}$ and it satisfies $f'(0,1,2)=1$, and we are back in the subcases considered above.

\item If $f(1, 0, 2)=0$ then  by applying a permutation  to the elements $0$ and $1$ in all operations of $\mathbb{A}$  we obtain $f'(0,1,2)=1$,  $r'(1,1,0,0,2)=2$ and $r'(1,0,1,0,2)=1$, with $r'$ the first projection on $\{0,1\}$. It follows that the operation 
$$s(x_1, \ldots, x_5)=r'(f'(x_1, x_4, x_5), x_3, x_2, x_1, x_5)$$ is the first projection on $\{0,1\}$ and it satisfies $s(0,0,1,1,2)=2$ and $s(0,1,0,1,2)=1$, so we are back in Case 1 considered above.

\item If  $f(2,0,1)= 0$, then we consider the possibilities for the operations $g$ and $h$. If $h(0,1,2)=1$ then we can just define an operation $f'=h$.
 If $h(1,0,2)=2$ (or $g(1,0,2)\in \{0,2\}$) then the operation
$f'(x,y,z)=f(h(x,y,z),y,x)$ (or $f(g(x,y,z), y, x)$) acts as the first projection on $\{0,1\}$ and it satisfies $f'(1,0,2)=0$. 
If $h(2,0,1)=1$ then  we set $f'=h$ and are back in the subcase considered above.
If  $h(0,1,2)=2$  we have the possibility that $\{1\}, \{0,2\}$ is a congruence of $\mathbb{A}$. To break this congruence we must have an operation $z$ in $\mathbb{A}$ that satisfies $z(0,2,0,2,1)\in \{0,2\}$,  $z(0,0,2,2,1)=1$, and is a projection on $\{0,1\}$

If $z$ is the first or second projection on $\{0,1\}$, we  define and operation $f'(x,y,z)=z(x,x,z,z, y)$ that is the first projection on $\{0,1\}$ and it satisfies $f'(0,1,2)=1$. If $z$ is the third or fourth projection on $\{0,1\}$ we define $f'(x,y,z)=z(y,y,x,x,z)$, this operation is the first projection on $\{0,1\}$ and it satisfies $f'(2,0,1)=1$.

If $z$ is the fifth projection of $\{0,1\}$ then the operation $g'(x,y,z)=z(y,z,y,z, x)$ is the first projection on $\{0,1\}$ and it satisfies $g'(1,0,2)\in \{0,2\}$.  
In all cases we reduced the problem to an already considered case.

\end{enumerate}

\subsubsection{Case 3} $r(1,1,0,0,2)=2$ and $r(0,1,0,1,2)\in \{0,1\}$ and $r$ is the first projection on $\{0,1\}$.

\begin{enumerate}
\item If $f(0,1,2)=1$ then  the operation $$s(x_1, \ldots, x_5)=r(f(x_1, x_3, r(x_3, x_4, x_1, x_2, x_5)), x_4, x_1, x_2, x_5)$$ is the first projection on $\{0,1\}$ and it satisfies $s(0,0,1,1,2)=2$ and $s(0,1,0,1,2)\in \{0,1\}$.

\item If $f(1,0,2)=0$ then permuting $0$ and $1$ in all operations of $\mathbb{A}$ we obtain operations that act as the first projection on $\{0,1\}$ and satisfy $f'(0,1,2)=1$, $r'(0,0,1,1,2)=2$ and $r'(1,0,1,0,2)\in \{0,1\}$. Then defining \[s(x_1, \ldots, x_5)=r'(f'(x_1, x_2, r'(x_3, x_1, x_4, x_2, x_5)), x_1, x_4, x_3, x_5)\] we know it acts as the first projection on $\{0,1\}$ and it satisfies $s(0,0,1,1,2)=2$ and $ s(0,1,0,1,2)\in \{0,1\}$.

\item If $f(2,0,1)=0$ then the operation $f'(x,y,z)=f(r(x,x,y,y,z,z),y, x)$ satisfies $f'(1,0,2)=0$ so we are back in the case considered just above.

\item If $f(2,0,1)=1$ then we look at the possibilities for the operation $g$. If $g(1,0,2)=0$  or $g(2,0,1)=0$ then we use $g$ like we used $f$ above, if $g(0,1,2)=2$ (or $h(0,1,2)=2$) then the operation $f'(x,y,z)=f(g(x,y,z), x, y)$ acts as the first projection on $\{0,1\}$ and it satisfies $f'(0,1,2)=f(g(0,1,2), 0,1)=f(2,0,1)=1$. This case was already considered above. Finally, if $g(1,0,2)=2$ we need to have another operation in $\mathbb{A}$ to break the congruence $\{1,2\}, \{0\}$, let it be $m$. This operation satisfies $m(1,2,1,2,0)\in \{1,2\}$, $m(1,1,2,2,0)=0$ and acts as a projection on $\{0,1\}$. If $m$ is the first or second projection on $\{0,1\}$ then the operation $f'(x,y,z)=m(x,x,z,z,y)$ is the firts projection on $\{0,1\}$ and it satisfies $f'(1,0,2)=0$. If $m$ acts as the third or fourth projection on $\{0,1\}$ then the operation $f'(x,y,z)=m(z,z,x,x,y)$ acts as the first projection on $\{0,1\}$ and it satisfies $f'(2,0,1)=0$. If $m$ acts as the fifth projection on $\{0,1\}$ then $h'(x,y,z)=m(y,z,y,z,x)$ acts as the first projection on $\{0,1\}$ and it satisfies $h'(0,1,2)\in \{1,2\}$. All these cases have already been considered.

\end{enumerate}

 \subsubsection{Case 4} $r(1,1,0,0,2)=2$ and $r(1,0,0,1,2)\in \{0,1\}$
 
 \begin{enumerate}
 \item If $f(0,1,2)=1$
then the operation $$s(x_1, \ldots, x_5)=r(f(x_1, x_4, x_5), x_3, x_1, x_2, x_5)$$ is the first projection on $\{0,1\}$ and it satisfies $s(0,0,1,1,2)=2$ and $s(0,1,0,1,2)\in \{0,1\}$.
 
\item If $f(1,0,2)=0$ then permuting $0$ and $1$ in all operations of $\mathbb{A}$ we obtain operations that act as the first projection on $\{0,1\}$ and satisfy $f'(0,1,2)=1$, $r'(0,0,1,1,2)=2$ and $r'(0,1, 1,0,2)\in \{0,1\}$. Then the operation $s(x_1, \ldots, x_5)=r'(x_1,x_2,x_4,x_3, x_5) $  acts as the first projection on $\{0,1\}$, it satisfies $s(0,0,1,1,2)=2$ and $ s(0,1,0,1,2)\in \{0,1\}$, so we are back in Case 1.

\item If $f(2,0,1)=0$  or $f(2,0,1)=1$ then this can be dealt with just like in  Case 3.

\end{enumerate} 
 
      \subsubsection{Case 5}  $r(0,0,1,1,2)=2$ and $r(1,0,0,1,2)\in \{0,1\}$ 
  \begin{enumerate}
 \item If $f(0,1,2)=1$   then the operation $$s(x_1, \ldots, x_5)=r( f(x_1, x_2, r(x_3, x_1, x_2, x_4, x_5)), x_1, x_3, x_4, x_5)$$   is the first projection on $\{0,1\}$ and it satisfies $s(0,0,1,1,2)=2$ and $s(0,1,0,1,2)\in \{0,1\}$. 
 
 \item   If $f(1,0,2)=0$ then permuting $0$ and $1$ in all operations of $\mathbb{A}$ we obtain operations that act as the first projection on $\{0,1\}$ and satisfy $f'(0,1,2)=1$, $r'(1,1,0,0,2)=2$ and $r'(0,1, 1,0,2)\in \{0,1\}$. Then defining $s(x_1, \ldots, x_5)=r'(x_1,x_2,x_4,x_3, x_5) $ we know it acts as the first projection on $\{0,1\}$ and it satisfies $s(1,1,0,0 2)=2$ and $ s(0,1,0,1,2)\in \{0,1\}$, and we are back in Case 3.

 \item If $f(2,0,1)=0$  or $f(2,0,1)=1$ then this can be dealt with just like in  Case 2.

 \end{enumerate}
   
   \subsubsection{Case 6}  
   
    $r(2,0,0,1,1)=2$ and $r(2,0,1,0,1)\in \{0,1\}$
    
\begin{enumerate}    
 \item If  $g(1,0,2)=2$ (or $h(1,0,2)=2$) the operation $$s(x_1, \ldots, x_5)=r( g(x_1, x_3, x_5),  x_3, x_4, x_2, x_1)$$   is the first projection on $\{0,1\}$ and it satisfies $s(1,1,0,0, 2)=2$ and $s(1, 0,0,1,2)\in \{0,1\}$;
 
 \item If  $g(0,1, 2)=2$ (or $h(0,1,2)=2$) the operation $$s(x_1, \ldots, x_5)=r( g(x_1, x_4, x_5),  x_1, x_2, x_3, x_4)$$   is the first projection on $\{0,1\}$ and it satisfies $s(0,0,1,1, 2)=2$ and $s(0,1,0,1,2)\in \{0,1\}$;

 \item We now assume that $g(1,0,2)=0$  or $g(2,0,1)=0$.
 We start by noting that we also have $h(0,1,2)=1$ or $h(2,0,1)=1$, so in particular $h$ satisfies $h(\{0,2\}, D, D)=h(D, \{0,2\}, D)=h(D, D, \{0,2\})=D$. 
Then we can obtain generalized Hubie-polymorphisms. First notice that:

\begin{claim}

$r(0,2,2,2,2)=2$
\end{claim}
\begin{proof}
Let us assume for a contradiction that $r(0,2,2,2,2)\in \{0,1\}$, then the operation 

$$h(g (r(x_1, \ldots, x_5),r(x_6, \ldots, x_{10}),r(x_{11}, \ldots, x_{15})), \ldots, g(r(x_{31}, \ldots, x_{35}), r(x_{36}, \ldots, x_{40}), r(x_{41}, \ldots, x_{45})))$$
is a Hubie-pol on $2$.
Indeed we have $r(2,D, D, D, D)\supseteq \{0,2\}$ or  $r(2,D, D, D, D)\supseteq \{1,2\}$, and, by the assumption,  
\[
\begin{array}{l}
r(D, 2, D, D, D), \cdots, r(D, D, D, D, 2)\supseteq \{0,2\}, \mbox{ or} \\
r(D, 2, D, D, D), \cdots, r(D, D, D, D, 2)\supseteq \{1,2\},
\end{array}
\] 
and $g(\{0,2\}, D, D)\supseteq \{0,2\}$, $g(D, \{0,2\}, D)=g(D, D, \{0,2\})=D$, as well as $g(\{1,2\}, D, D)=g(D, \{1,2\}, D)=g(D, D, \{1,2\})=D$, and $h(\{0,2\}, D, D)=h(D, \{0,2\}, D)=h(D, D, \{0,2\})=D$.
\end{proof}\qed

\begin{claim}
The $45$-ary operation $c_1$ defined as 
$$\begin{array}{rl}r(& 
 h(g(x_1, \ldots, x_3),g(x_4, \ldots, x_6), g(x_{7}, \ldots, x_{9}) ),\\
 &  \cdots,\\
 &   h(g(x_{37}, \ldots, x_{39}),g(x_{40}, \ldots, x_{42}),g(x_{43}, \ldots, x_{45})) \ )\end{array}$$
 is a generalized Hubie-pol on the element $ 102 (\times 15)$ if $g(1,0,2)=0$ or $201( \times  15)$ if we assume that $g(2,0,1)=0$, returning $0$ on this tuple.
\end{claim}

\begin{proof}
When applying $c_1$ to the tuple above, recalling that $g$, $h$ and $r$ are idempotent, we obtain 
 $$h(g(1,0,2), g(1,0,2), g(1,0,2))=h(0,0,0)=0 \ {\rm or } $$
 $$ h(g(2,0,1), g(2,0,1), g(2,0,1))=h(0,0,0)=0.$$
 Now let us check that the operation is a generalized Hubie-pol: we have  $g(1, D, D), g(D, 0, D),$ $g(D, D, 1)\supseteq \{0,1\}$, $g(2, D, D), g(D, D, 2)\supseteq \{0,2\}$, $h(\{0,2\}, D, D)=h(D, \{0,2\}, D)=h(D, D, \{0,2\})=D$,  and   $r(\{0,1\}, D, D, D, D)=\cdots =r(D, D, D, D, \{0,1\})=D$ (see Claim above). This proves the claim.
\end{proof}\qed

\begin{claim}
The $45$-ary operation $c_2$ defined as 
  $$\begin{array}{rl}
  g(& h(r(x_1, \ldots, x_5), r(x_6, \ldots, x_{10}), r(x_{11}, \ldots, x_{15}))\\
  & \cdots\\
  &  h(r(x_{31}, \ldots, x_{35}), r(x_{36}, \ldots, x_{40}), r(x_{41}, \ldots, x_{45})) \ )
  \end{array}$$
s a generalized Hubie-pol  on the elements 
$20011(\times 9)$ and on this tuple it returns $2$.
\end{claim}
\begin{proof}
When applying $c_2$ to the tuple above, recalling that $g$, $h$ and $r$ are idempotent, we obtain 
 $$h(r(2,0,0,1,1),r(2,0,0,1,1),r(2,0,0,1,1))=h(2,2,2)=2.$$

 Now let us check that the operation is a generalized Hubie-pol: we have $r(2, D, D, D, D)\supseteq \{0,2\}$ or $r(2, D, D, D, D)\supseteq \{1,2\}$ and $r(D, 0, D, D, D)= r(D, D, 0, D,D)=r(D, D, D, 1, D)= r(D, D, D, D, 1)=D$. We also have  $h(\{0,2\}, D, D)=h(D, \{0,2\}, D)=h(D, D, \{0,2\})=D$ and $g(\{1,2\}, D, D)=h(D, \{1,2\}, D)=h(D, D, \{1,2\})=D$, this proves the claim.
 \end{proof}\qed

\begin{claim}
The $45$-ary operation $c_3$ defined as 
$$\begin{array}{rl}r(& 
 g(h(x_1, \ldots, x_3),h(x_4, \ldots, x_6), h(x_{7}, \ldots, x_{9}) ),\\
 &  \cdots,\\
 &   g(h(x_{37}, \ldots, x_{39}),h(x_{40}, \ldots, x_{42}),h(x_{43}, \ldots, x_{45})) \ )\end{array}$$
 is a generalized Hubie-pol on the elements $ 012(\times 15)$ if $h(0,1,2)=1$, or on the element $201(\times 15)$ if we assume that $h(2,0,1)=1$, returning $1$ on this tuple.
\end{claim}
\begin{proof}
When applying $c_3$ to the tuple above, recalling that $g$, $h$ and $r$ are idempotent, we obtain 
 $$g(h(0,1,2), h(0,1,2), h(0,1,2))=g(1,1,1)=1 \ {\rm or } $$
 $$ g(h(2,0,1), h(2,0,1), h(2,0,1))=g(1,1,1)=1.$$
 Now let us check that the operation is a generalized Hubie-pol: we have  $h(0, D, D),h (D, 0,D),$ $h(D, 1, D), h(D, D, 1)\supseteq \{0,1\}$, $h(2, D, D), h(D, D, 2)\supseteq \{1,2\}$, $g(\{1,2\}, D, D)=g(D, \{1,2\}, D)=g(D, D, \{1,2\})=D$,  and   $r(\{0,1\}, D, D, D, D)=\cdots =r(D, D, D, D, \{0,1\})=D$ (see Claim above). This proves the claim.
 \end{proof}\qed

\end{enumerate}

\subsection{$\mathbb{A}$ has congruences but they do not yield $G$-sets} 

If $\{0,1\}, \{2\}$ is the kernel of a congruence then there must  exist an operation $z$ on the two element algebra  with domain $\{ \{0,1\}, \{2\}\}$ that acts as either majority, minority, meet, or join. 
We must get similar operations if $\{0,2\}, \{1\}$ or $\{1,2\}, \{0\}$ are congruences. Suppose that $l$ is an operation on the two element domain $\{\{0,2\}, \{1\}\}$ that acts as either majority, minority, or semilattice. When extending $l$ to $\mathbb{A}$ we obtain that $l(0,1,1)=1$if $l$ is a majority, $l(0, 0, 1)=1$ is $l$ is a minority, $l(0,1)=1$ or $l(1,0)=0$ if $l$ is a semilattice operation.  All these options contradict the fact that $\{ 0,1\}$ is a $G$-set, hence  $\{0,2\}, \{1\}$, and in a similar way  $\{1,2\}, \{0\}$, cannot be congruences of $\mathbb{A}$.

We look at the different possibilities for the operation $z$:

\subsubsection{$z$ is a majority}

then extending $z$ to $\mathbb{A}$ we must have $z(2,2,x)=z(2,x,2)=z(x,2,2)=2$ and $z(x,y, 2)=z(2,x,y)=z(x,2,y)\in \{x,y\}$ for any $x, y\in \{0,1\}$ and $z$ acts as a projection on $\{0,1\}$, we assume wlog that it is the first projection.

Then, the 
$9$-ary operation
$$f(z(x_1,x_2,x_3),z(x_4,x_5,x_6),z(x_7,x_8,x_9))$$ is a Hubie-pol on $\{0\}$ if $f(0,1,2)=1$ or $f(2,0,1)=1$, and is a Hubie-pol on $\{1\}$ if $f(1,0,2)=0$ or $f(2,0,1)=0$.

\subsubsection{$z$ is a minority}

then extending $z$ to $\mathbb{A}$ we must have $z(2,2,x)=z(2,x,2)=z(x,2,2)\in \{x,y\}$ and $z(x,y, 2)=z(2,x,y)=z(x,2,y)=2$ for any $x, y\in \{0,1\}$ and $z$ acts as a projection on $\{0,1\}$. 

Then, as above,  the $9$-ary operation
$$f(z(x_1,x_2,x_3),z(x_4,x_5,x_6),z(x_7,x_8,x_9))$$ is a Hubie-pol on $\{0\}$ if $f(0,1,2)=1$ or $f(2,0,1)=1$, and is a Hubie-pol on $\{1\}$ if $f(1,0,2)=0$ or $f(2,0,1)=0$.

\subsubsection{$z$ is join semilattice}
from  $z(\{0,1\}, \{2\})=z(\{2\}, \{0,1\})=\{2\}$, extending $z$ to $\mathbb{A}$ we obtain $z(x,2)=z(2,x)=2$ for any $x\in \{0,1\}$ and $z$ is, wlog, the first projection on $\{0,1\}$.
Then the $6$-ary operation
$$f(z(x_1,x_2),z(x_3,x_4,),z(x_5,x_6))$$ is a Hubie-pol on $\{0\}$ if $f(0,1,2)=1$ or $f(2,0,1)=1$, and is a Hubie-pol on $\{1\}$ if $f(1,0,2)=0$ or $f(2,0,1)=0$.

\subsubsection{$z$ is meet semilattice}
from  $z(\{0,1\}, \{2\})=z(\{2\}, \{0,1\})=\{0,1\}$, extending $z$ to $\mathbb{A}$ we obtain $z(x,2)=z(2,x)\in \{0,1\}$ for any $x\in \{0,1\}$. 

\begin{enumerate}
\item If $z(2, 0)=z( 2,1)=0 $ then  the $18$-ary  operation 
$$h( f(z(x_1, x_2),\ldots, z(x_{5}, x_{6})), \ldots, f(z(x_{13}, x_{14}), \ldots, z(x_{17}, x_{18}) ))$$
is a Hubie-pol on $\{2\}$ whenever $f(0,1,2)=1$,  $f(1,0,2)=0$,  or $f(2,0,1)=1$,  and $h(0,1,2)=1$ or $h(2,0,1)=1$. Note that if $f(1,0,2)=0$ then, by permuting $0$s and $1$s is all operations of $\mathbb{A}$ we obtain $f(0,1,2)=1$, and this permutation does not affect $z$. 

Assume now that we have $f(2,0,1)=0$. If $h(1,0,2)=2$ then the operation $f'(x,y,z)=f(h(x,y,z), y, x)$ is the first projection on $\{0,1\}$ and it satisfies $f'(1,0,2)=0$ so we are back in a previous case.
If  $h(0,1,2)=2$  then we must also have an operation in $\mathbb{A}$ that breaks the congruence $\{0,2\}, \{1\}$. We have seen in part 4 of Case 2 considered above that we can then reduce this case to another one previously considered.

 \item If $z(2, 0)=z( 2,1)=1 $ then  the $18$-ary  operation 
$$g( f(z(x_1, x_2),\ldots, z(x_{5}, x_{6})), \ldots, f(z(x_{13}, x_{14}), \ldots, z(x_{17}, x_{18}) ))$$
is a Hubie-pol on $\{2\}$ whenever $f(0,1,2)=1$,  $f(1,0,2)=0$,  or $f(2,0,1)=0$,  and $g(1,0,2)=0$ or $g(1,0,2)=2$. Note that if $f(0,1,2)=1$ then, by permuting $0$s and $1$s is all operations of $\mathbb{A}$ we obtain $f(1,0,2)=0$, and this permutation does not affect $z$. 

Assume now that we have $f(2,0,1)=1$. If $g(0,1, 2)=2$ then the operation $f'(x,y,z)=f(g(x,y,z), y, x)$ is the first projection on $\{0,1\}$ and it satisfies $f'(0,1,2)=1$ so we are back in a previous case.
If  $g(1,0 ,2)=2$  then we must also have an operation   that breaks the congruence $\{1,2\}, \{0\}$.  We have seen in part 4 of Case 3 considered above that we can then reduce this case to another one previously considered.

\item If $z(2, D)=D$ then $z$ is a Hubie-pol on $2$.
\end{enumerate}

\section*{Appendix C: A three-element vignette}

\noindent \textbf{Theorem~\ref{thm:vignette}}.
Let $\mathbb{A}$ be an idempotent algebra on a $3$-element domain. Either 
\begin{itemize}
\item $\Pi_k$-CSP$(\mathrm{Inv}(\mathbb{A}))$ is in NP, for all $k$; or
\item $\Pi_k$-CSP$(\mathrm{Inv}(\mathbb{A}))$ is co-NP-complete, for all $k$; or
\item $\Pi_k$-CSP$(\mathrm{Inv}(\mathbb{A}))$ is $\Pi^{\mathrm{P}}_2$-hard, for some $k$.
\end{itemize}

\begin{proof}
If $\mathbb{A}$ has PGP then it is Switchable and QCSP$(\mathrm{Inv}(\mathbb{A}))$ is in NP from Theorem~\ref{thm:easy}. It follows that $\Pi_k$-CSP$(\mathrm{Inv}(\mathbb{A}))$ is in NP, for all $k$, \emph{a fortiori}. Suppose now that $\mathbb{A}$ has EGP. 

If $\mathbb{A}$ does not contain a G-set as a factor, then $\mathbb{A}$ generates the  semilattice-without-unit $s$ and it is known that $\Pi_k$-CSP$(\mathrm{Inv}(\mathbb{A}))$ is in co-NP, for all $k$ \cite{HubieExRes}. Since we have co-NP-hardness from Theorem~\ref{thm:hard}, we can indeed in this case upgrade to co-NP-completeness, for all $k$.  

We now assume that $\mathbb{A}$ contains a G-set as a factor. If $\mathbb{A}$ is a G-set then an examination of the proof in \cite{BBCJK} will show that $\Pi_2$-CSP$(\mathrm{Inv}(\mathbb{A}))$ is already $\Pi^{\mathrm{P}}_2$-hard. This is because additional auxiliary existential quantification may always be pushed innermost (see Proposition~7 of \cite{BovaChenC14}). More generally, if $\mathbb{A}$ has a G-set as a homomorphic image then $\Pi_2$-CSP$(\mathrm{Inv}(\mathbb{A}))$ is $\Pi^{\mathrm{P}}_2$-hard (see Lemma~5 from \cite{QCSPmonoids}). Thus, since $\mathbb{A}$ is over three elements, we can assume the remaining case is that $\mathbb{A}$ has a $2$-element G-set as a subalgebra.

Since $\mathbb{A}$ has EGP, there exist $\alpha,\beta$ strict subsets of $A$ so that $\alpha \cup \beta=A$ and all operations of $\mathbb{A}$ are $\alpha \beta$-projective. If $\alpha \cap \beta = \emptyset$ then $\mathbb{A}$ has a $2$-element G-set as a homomorphic image (with $\alpha$ and $\beta$ the two equivalence classes) and we are in a previous case. Let us assume \mbox{w.l.o.g.} that $\alpha:=\{0,2\}$ and $\beta:=\{1,2\}$.

Case A. G-set is on $\{0,1\}$. We will argue that the co-NP-hardness proof of Theorem~\ref{thm:hard} can be readily extended to $\Pi^{\mathrm{P}}_2$-hardness already for the $\Pi_2$-CSP$(\mathrm{Inv}(\mathbb{A}))$. Observe that $\{0,1\}$ is a subalgebra and therefore is in Inv$(\mathbb{A})$. To add existential quantification $\exists v$, for a reduction from the complement of 3-$\Pi_2$-NAESAT, we simply need to add the stipulation $v \in \{0,1\}$, which appears, as everything else, in DNF.


Case B. G-set is on $\{0,2\}$. Whereas in Case A we extended by alternation the co-NP-hardness proof that used $0$ and $1$ to indicate true and false, we will here extend by alternation the NP-hardness proof of CSP$(\mathrm{Inv}(\mathbb{A}))$ that arises from $\{0,2\}$ inducing a G-set. Thus, we will use $0$ and $2$ to represent true and false. We first make the crucial observation that the binary relation $Z:=\{(0,0),(2,1),(2,2),(0,2)\}$ is in Inv$(\mathbb{A})$. To see this, imagine a term operation $f$ of $\mathbb{A}$ which is both  $\alpha \beta$-projective and actually projective on tuples from $\{0,2\}$ (recall this induces a G-set). Furthermore, this must be the same co-ordinate that is being projected upon for both of these since if this co-ordinate is a $2$, the outcome must be a $2$ already by $\alpha\beta$-projectivity. Since $Z$ has only four pairs we may assume $f$ is at most $4$-ary and we can consider its action columnwise on
\[
\begin{array}{cc}
f & f \\
0 & 0 \\
2 & 1 \\
2 & 2 \\
0 & 2 \\
\hline
x & y. \\
\end{array}
\]
If $f$ ($\alpha \beta$-)projects to the first co-ordinate, then we have $(x,y):=(0,0)$ or $(0,2)$. If it ($\alpha \beta$-)projects to the second co-ordinate, then we have $(x,y):=(2,1)$ or $(2,2)$. If it ($\alpha \beta$-)projects to the third co-ordinate, then we must have $(x,y):=(2,2)$. Finally, if it ($\alpha \beta$-)projects to the fourth co-ordinate then we can have only $(x,y):=(0,2)$.

Now, let us imagine a reduction from 3-$\Pi_2$-NAESAT where $0$ and $2$ will represent true and false. Note that the ternary predicate $R:=\{0,2\}^3\setminus\{(0,0,0),(2,2,2)\}$ is in Inv$(\mathbb{A})$. We will use $R$ to enforce the not-all-equal predicate on $\{0,2\}$ in the obvious fashion and the existential variables from 3-$\Pi_2$-NAESAT will become existential variables of $\Pi_2$-CSP$(\mathbb{A})$ restricted to be from $\{0,2\}$ which, as a subalgebra, is in Inv$(\mathbb{A})$. The trick is how to encode universal variables $\forall v$ and for this we augment a new auxiliary variable $v'$ and substitute by $\forall v' \exists v \ Z(v,v')$. When $v'$ is evaluated as $1$, $v$ is forced to be $2$; when $v'$ is evaluated as $0$, $v$ is forced to be $0$; and when $v'$ is evaluated as $2$, $v$ can be either $0$ or $2$. Ostensibly this does not result in an instance of $\Pi_2$-CSP$(\mathrm{Inv}(\mathbb{A}))$ until we notice, as per the previous sentence, that the existential quantification of all the auxiliary variables may be pushed innermost.

Case C. G-set is on $\{1,2\}$. This case is symmetric with Case B.
\end{proof}

\end{document}